\documentclass[12pt]{amsart}
\usepackage{graphicx}
\usepackage{geometry}
\usepackage{amssymb}
\usepackage{enumerate}
\usepackage{setspace}
\usepackage[dvipsnames]{pstricks}
\usepackage{pst-node}
\usepackage{pst-slpe}
\usepackage{pst-tree}
\usepackage{float}
\usepackage{amsfonts}
\usepackage{amsmath}
\usepackage[round,comma]{natbib}
\usepackage{amsthm}
\usepackage{lmodern}
\usepackage[T1]{fontenc}
\usepackage{mathtools}
\usepackage{ulem}
\usepackage{tikz}
\usepackage{subcaption}
\usepackage[flushleft]{threeparttable}
\usetikzlibrary{arrows,shapes,trees,backgrounds,calc,through}
\usepackage{hyperref}
\hypersetup{
    colorlinks=false,
    pdfborder={0 0 0},
}

\newcommand{\notr}{\!\not\!\! R}
\newcommand{\supp}{\text{supp}}
\newcommand{\consequence}{\pi}
\newcommand{\vari}{Y}
\newcommand{\covariate}{Q}
\DeclareMathOperator{\marg}{marg}

\theoremstyle{definition}
\newtheorem{defn}{\protect\definitionname}
\theoremstyle{plain}

\theoremstyle{plain}

\theoremstyle{definition}
\newtheorem{ax}{\protect\axiomname}
\newtheorem*{ax*}{\protect\axiomname}

\theoremstyle{plain}
\newtheorem{thm}{\protect\theoremname}
\theoremstyle{plain}
\newtheorem{lemma}{Lemma}
\theoremstyle{plain}
\newtheorem{cor}{\protect\corollaryname}
\theoremstyle{plain}
\newtheorem{proposition}{\protect\propositionname}
\theoremstyle{remark}
\newtheorem{rem}{Remark}
\theoremstyle{definition}

  \providecommand{\axiomname}{Axiom}
  \providecommand{\conjecturename}{Conjecture}
  \providecommand{\definitionname}{Definition}
  \providecommand{\corollaryname}{Corollary}
  \providecommand{\theoremname}{Theorem}
  \providecommand{\propositionname}{Proposition}

\begin{document}
\title[Subjective Causality in Choice]{Subjective Causality in Choice$^*$}

\author[Ellis and Thysen]{Andrew Ellis$^\dag$ and Heidi Christina Thysen$^\S$}

\date{May 2024}
\thanks{$^*$
We thank Sandro Ambuehl, Kfir Eliaz, Yusufcan Masatlioglu, Jay Lu, Pablo Schenone,  Heiner Schumacher, Balazs Szentes, Chen Zhao, and especially Rani Spiegler, as well as participants at the Arizona, BRIC, BU, DC Area Theory Workshop, Edinburgh, Israeli and Hong Kong Theory Seminars, Harvard/MIT, Michigan, LSE, RUD, St Andrews, UCLA, and UCSD  for helpful comments and discussions. Thysen gratefully acknowledges financial support by the ERC Advanced Investigator grant no. 692995.}
\thanks{$^\dag$Department of Economics, London School of Economics and Political Science,   a.ellis@lse.ac.uk.}
\thanks{$^\S$Department of Economics, Norwegian School of Economics (NHH),  heidi.thysen@nhh.no.}

\begin{abstract}
Choices based on observational data depend on beliefs about which correlations  reflect causality. An agent predicts the consequence of available actions   using a dataset and her subjective beliefs about causality represented by a directed acyclic graph (DAG). We identify her DAG from her random choice rule. Her choices reveal the chains of causal reasoning that she undertakes and the confounding variables she adjusts for, and these pin down her model. When her choices generate the dataset used, her behavior affects her inferences, which in turn affect her choices.  We provide necessary and sufficient conditions for testing whether her behavior is compatible with such a model.
\end{abstract}
\maketitle

\newpage
\section{Introduction}
A person's choice depends on what she predicts its consequences will be. When observational data informs these predictions, her choice depends on her beliefs about which correlations in the data reflect causality and which do not. While correlation is unambiguous, inferences about causality are inherently subjective.
For instance, a positive correlation between duration of hospitalization and risk of death may be causal if time in hospital increases one's risk of catching an unrelated infection, or spurious if both are caused by the severity of illness.
These beliefs about causality affect behavior:
she would be more reluctant to seek treatment if she believed the former rather than the latter.
This paper develops a theory under which we  can use a person's choices to identify her subjective beliefs about causality and to test whether those beliefs can explain her behavior.

Following \cite{Pearl2009} and \cite{Spiegler2016}, we study a decision maker (DM) who makes predictions about the consequences of her action using a causal model described by a Directed Acyclic Graph (DAG). A DAG is a graph where each node represents a variable, and each edge represents a belief that one variable directly causes another. For example,  an edge  from  ``Disease''  to ``Hospitalization''    represents the belief that going to the hospital  is a consequence of being sick.  A variable whose node has no edges pointing into it is believed to be exogenous.
DAGs allow for a flexible and non-parametric representation of causal relationships. 

How someone thinks about causality is inherently unobservable.
We instead consider a random choice rule that records how often the DM takes each action after observing each dataset describing the joint distribution of variables and actions.
The choice rule has a \emph{subjective causality representation (SCR)} if she uses a fixed DAG to predict the consequences of her actions, and then chooses the action with the highest predicted utility most frequently. That is, the DM estimates the strength of the causal relationships in her DAG using the dataset, makes predictions about the outcomes of every action, and chooses the best.   We show how to reveal the causal model of any choice rule with an SCR, and provide axiomatic foundations for an SCR with a dataset that is endogenously generated by the DM's choices.

%
%

To illustrate, consider a management consultant  who recommends the workforce size,
tiny ($a$) or big ($b$), that she thinks will maximize profits. To aid her decision, she has access to data detailing the firm's previous workforce sizes, its competitor's output (which can be high or low), and its profits (which can be large or small). 
The consultant
 infers expected profit from each workforce size using the dataset described in Table \ref{tab: example data}, and  recommends
 the size that she predicts will maximize profit.
We observe what the consultant recommended, and want to understand \emph{why} she made her recommendation. For instance, a manager at a different firm may want to  hire the consultant only if she will support the manager's preferred action given the second firm's circumstances.

\begin{table} 
\caption{A Dataset \label{tab: example data}}
\begin{threeparttable}
\begin{tabular}{|l|cc|c|c|}
\hline 
 Profits & Pr(\text{large \& $\cdot$}) & Pr(\text{small \& $\cdot$}) & Pr(\text{large$|\cdot$}) & Pr(\text{large$|$workforce}) \\
\hline
\text{$b$, high  output} & 0.24 & 0.16 & 0.6 & 0.52\\
\text{$a$, high  output} & 0.07 & 0.03 & 0.7 & 0.38\\
\text{$b$, low  output}  & 0.02 & 0.08 & 0.2 & 0.52\\
\text{$a$, low  output}  & 0.12 & 0.28 & 0.3 & 0.38\\

\hline  
\end{tabular} 

\begin{tablenotes}
      \small
      \item Table describes a simulated dataset. $a$ and $b$ are tiny and big workforce. Column  2 (3) shows the fraction of instances with large (small) profits that have the characteristics in Column 1. Column 4 (5) shows fraction of high profits conditional on both characteristics (only workforce) in Column 1.
    \end{tablenotes} 
\end{threeparttable}
\end{table}

Suppose that we know  that the consultant's subjective causal model is either $R$ or $R'$ in Figure \ref{fig:intro example DAGs}. Observing her recommendation from the dataset given by Table \ref{tab: example data} enables us to determine which she believes. If the consultant's model is $R$, then she believes that competitor output and  profit are both endogenous consequences of the size of the workforce; that is, she believes that the firm is a Stackelberg leader. The unconditional correlation between workforce size and profit is positive, so she is more likely to recommend a big workforce. If her model is instead $R'$, then she believes that competitor output is an exogenous cause of both workforce and profit; that is, the firm acts as a Stackelberg follower. Then, she recommends a tiny workforce more frequently because the correlation between workforce size and profit is negative conditional on either level of competitor output. 

\begin{figure}[h]
\caption{Two Subjective Causal Models}
	\begin{minipage}{0.4\textwidth}
	\centering
\begin{tikzpicture}[scale=1.2]
\node (v0) at (-1.5,0) {Workforce};
\node (vA) at (0,1) {Competitor Output};
\node (vP) at (1.5,0) {Profit};
\draw [->] (v0) edge (vA);
\draw [->] (vA) edge (vP);
\draw [->] (v0) edge (vP);
\end{tikzpicture}
	\vspace{-2em}
	\subcaption[first caption.]{$R$}
	\end{minipage}
	\qquad\qquad
	\begin{minipage}{0.4\textwidth}
	\centering
\begin{tikzpicture}[scale=1.2]
\node (v0) at (-1.5,0) {Workforce};
\node (vA) at (0,1) {Competitor Output};
\node (vP) at (1.5,0) {Profit};
\draw [->] (vA) edge (v0);
\draw [->] (vA) edge (vP);
\draw [->] (v0) edge (vP);
\end{tikzpicture}
	\vspace{-2em}
	\subcaption[first caption.]{$R'$}
	\end{minipage}
\label{fig:intro example DAGs}
\end{figure}
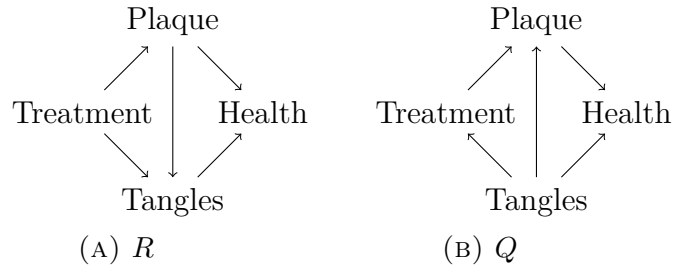


Our first main result shows that the DM's choices jointly reveal  her subjective causal model, the variable  she cares about (the outcome), and her preferences. 
Three features of the causal model determine her predictions.
Most crucially, her predictions depend on the relationships between the variables that appear in the simplest causal paths through which she believes that her action can influence the outcome. 
For instance, the DAG $R$ contains two causal paths from her action to the outcome: $\text{Workforce }\rightarrow \text{ Output } \rightarrow \text{ Profit}$ and $\text{Workforce }\rightarrow \text{ Profit}$. As the nodes in the latter causal path are a subset of the nodes in the former, changing the distribution of competitor output will not change the  DM's behavior,  \emph{ceteris paribus}. 
Her predictions also depend on the exogenous variables that have a common consequence  with the action, such as competitor output in the DAG $R'$. We refer to such variables as \emph{confounders}.
The DM adjusts her predictions to account for spurious correlations resulting from the confounders.
Any two causal models with the same confounders and the same simplest causal paths from either the action or a confounder to the outcome lead to the same predictions for every dataset.  This is reminiscent of Occam's razor, in that variables uninvolved in any of the simplest causal paths can be safely ignored.
We show that one can reveal these features by observing the DM's choices given different datasets. When we have additional knowledge about the DM's choice procedure,  then the variation required in the datasets decreases. This provides potential avenues to experimentally elicit people's subjective causal models.\footnote{Of course, one will have to consider the nature of the choice environment. This procedure is suitable for situations in which the people have formed an attachment to their causal model.}

In contrast to the large literature that empirically determines causality (e.g. \cite{Card1999}), the result identifies an agent's \emph{perception} of causal relationships. A policy change that does not objectively change the effectiveness of the agent's action might still change her prediction about its effect and thereby her choices. For instance, a firm that appears reluctant to hire minority workers may dislike employing minorities even though they are more or equally productive (taste-based discrimination). Alternatively, its reluctance may arise because minority-status is correlated with another attribute, such as education, that the firm thinks affects productivity (statistical discrimination, perhaps based on a wrong causal model, resulting in incorrect beliefs). Policies that attempt to remedy the latter, such as affirmative action in university admissions or awarding scholarships to minority students, may do nothing for the former.\footnote{
See \cite{Lang2020}  for an overview of different types of race discrimination, and \cite{Bohren2023} for an overview the limited extend to which the literature attempts to take misspecified beliefs into account.} Thus, knowing the firm's causal model may help determine the most effective policy to mitigate discrimination.

We extend our model to datasets generated by the DM's own (or others with the same model's) past behavior. Actions correspond to lotteries over the other variables. Her choices from each set  generate a dataset that combines the choice frequency of lotteries with their resulting distributions over the other variables. 
The choice rule has an \emph{Endogenous SCR} if she forms predictions using this endogenous dataset and her subjective causal model, and then chooses the lottery with the highest predicted expected utility plus a (extreme-valued) shock.
Her behavior influences the correlation between the variables, which in turn influence her predictions and thus her choices. While this feedback occurs in many studies of agents with misspecified models, e.g. \cite{EspondaPouzo2016Berk}, it has typically been absent in decision theoretic work on misspecification. We show in Sections \ref{sec: reg violation} and \ref{sec:self-confirming} that the feedback effect allows the model to accommodate a number of documented cognitive biases, including selection neglect, illusion of control, status-quo, and congruence.

 This feedback effect may, for example, cause a firm to incorrectly statistically discriminate against minority workers when it hires few of them, but not when it hires more.  Consider a firm who thinks that worker productivity is a function of only education. If education is correlated with productivity for majority workers, but not for minorities, then the perceived return to education increases with the fraction of majority workers hired. This incentivizes the firm to hire the better-educated type more frequently, potentially reinforcing the effect. While this allows the model to accommodate the above biases, it lead to technical challenges for our analysis, such as violations of (the stochastic choice analog of) independence of irrelevant alternatives, a necessary condition for a random utility model (RUM).

Our second main result establishes how to test whether a random choice rule over lotteries has an  Endogenous SCR. We provide axioms that link the DM's subjective causal model, as inferred by our first results, to her behavior. The axioms require that if her predictions about the outcome for a pair of actions are constant across two menus, then their relative choice frequency is also constant. Her choices from a pair of menus diverge from the predictions of Logit with an expected utility Luce index (Logit-EU) only if her inferences about causal effects  differ across the menus. 

For a DM with an Endogenous SCR, the predicted distribution of the consequences from an action might not align with the one that actually results from taking it. 
The result establishes that subjective causality provides enough discipline on how her beliefs are distorted to be testable; without any restrictions on belief distortion, testing would be impossible.  More broadly, this paper adapts decision theoretic methodology to identify and test an agent's subjective model of the world and her preferences, as opposed to the usual exercise of identifying and testing her preferences with a correct, or at least agreed upon, model of the world. We see this as a step towards providing testable implications for the growing literature studying agents with misspecified models, especially \cite{Spiegler2016},  
\cite{Eliaz2020}, \cite{Spiegler2020a}, and \cite{Schumacher}, which all use versions of the subjective causality representation.%
\footnote{\label{fn:misspecifed}Other models where misspecification leads to distorted beliefs include \cite{EspondaPouzo2016Berk}, \cite{BohrenHauser2018}, \cite{Fricketal2019}, \cite{He2018}, \cite{HeidhuesKoszegiStrack2018},  \cite{SammuelsonMailath2019}, \cite{olea2021competing}, and \cite{Levyetal2021}.}

\subsection*{Related literature}

\cite{Pearl1995Causal} argued for using and analyzing DAGs to understand causality.
A large literature \citep[e.g.,][]{cowell1999,Pearl2009} develops and applies this approach for probabilistic and causal inference.
Applied researchers  use DAGs to estimate the causal effect of  interventions from  observational datasets, e.g. \cite{Tennat2020DAGHealth}.%
\footnote{ \cite{Imbens2020} contrasts this with the potential outcomes approach and discusses why these methods have attracted more attention outside of economics than within it.}  

\cite{Spiegler2016} first modeled misspecified causal reasoning using DAGs, albeit without axiomatic foundations.
He illustrates that his model has the power to capture a number of errors in reasoning, including reversed causality and omitted variables.
Taken together, our results allow one to test the underlying assumptions of existing work studying the effects of causal misperception.
This growing literature has been applied to  monetary policy \citep{Spiegler2020a}, political competition \citep{Eliaz2018model, eliaz2022false}, communication \citep{Eliaz2019}, inference \citep{Eliaz2020}, and contracting \citep{Schumacher}.
In both these papers and the present one,  the DM's behavior results from using a DAG and observational data to predict the outcome of her actions.
Consequently, our results identifying the DAG from behavior increase their applicability.


We follow Spiegler in focusing on how a DM with a subjective causal model interprets data and on the feedback between her predictions and behavior. In contrast,
\cite{Schenone2020causality} takes a more normative approach to causality. He considers a DM who expresses preferences over act-causal-intervention pairs, and provides necessary and sufficient conditions for her beliefs to result from applying the do-operator  to intervened variables for a fixed DAG and fixed prior. 
His analysis complements ours, showing what is possible given rich variation in exogenous interventions and mappings from variables to payoffs as opposed to rich variation in the data provided about the available actions. In contrast to both these approaches, \cite{Alexander} model perceived causal relationships as reducing Kolmogorov complexity but do not explicitly relate their notion to choice.

More generally, our paper is related to the decision theory literature studying DMs who misperceive the world.
\cite{Lipman1999} studies a DM who may not understand all the logical implications of information provided to her.
\cite{EllisPiccione2017} develop a model where agents misperceive the correlation between actions.
\cite{Kochov2015} models an agent who does not accurately foresee the future consequences of her actions.
\cite{Keetal2020} study DMs who perceive lotteries through a neural network.
\cite{EllisMasatlioglu2021} consider an agent who categorizes alternatives based on the context, and the alternative's categorization affects her evaluation (or perception) thereof.
\cite{Cerreiaetal2017ModelMisspecification} analyze a DM who considers several misspecified models and makes decisions that take her lack of confidence in the models into account.
In all, the DM's perception of an alternative is unaffected by her behavior.
 
Finally, our paper also falls into the theoretical literature studying random choice.
We fall between two strands.
The first seeks to use choice data to identify features of otherwise rational behavior, such as \cite{Gul2006Random} identifying the distribution of utility indexes, \cite{Lu2016} identifying an agent's private information, and \cite{apesteguia2018monotone} studying comparative risk and time preferences.
The second interprets randomness as a result of boundedly rational behavior in abstract environments, such as the \cite{manzini2014stochastic}, \cite{brady2016menu}, and \cite{cattaneo2020randomattention} models of limited attention.
This paper uses random choice to identify features of explicitly boundedly-rational behavior.
The only other paper of which we are aware that uses stochastic choice to capture equilibrium behavior is \cite{chambersetal2021influence}.

\section{Model}
\label{sec:Model}
\subsection{Primitives}
A DM chooses an action after observing a dataset $q$.
A set $\mathcal{X}_0$ having at least two elements denotes the set of possible actions.
Taking an action affects the distribution of a random vector $X=(X_1,\dots,X_{n})$ taking values in 
$\prod_{i=1}^{n} \mathcal{X}_{i}\equiv \mathcal{X}_{-0}$, where $\mathcal{X}_i$ is a finite subset of $\mathbb{R}$ with at least two distinct elements for each $i>0$.
The dataset $q$ informing her choice  describes past (joint) realizations of $X$ and actions, with $q(a,x_1,\dots,x_{n})$ representing frequency of observations of action $a$ and $X_i=x_i $ for every $i$.
Formally, $q$ is a distribution over $\mathcal{X} \equiv\mathcal{X}_0 \times \mathcal{X}_{-0}$ assigning positive probability to a finite set $S_q$ of actions and where, to avoid issues of updating on zero probability events, $ q(a,x)>0 $ for all $x \in \mathcal{X}_{-0}$ and $a \in S_q$; let $\mathcal{D}$ be the set of such datasets.
We also make the simplifying assumption that the DM's choice set equals $S_q$.%
\footnote{This avoids the DM having to form predictions about actions not in the dataset. One can easily extend the model to choice from a known  subset of $S_q$.}
An important special case, formally studied in Section \ref{sec:Foundation}, is when the dataset $q$ is derived endogenously from the DM's (random) choice from a set of lotteries over $\mathcal{X}_{-0}$.
	
A random choice rule is a function $\rho: \mathcal{X}_0 \times \mathcal{D} \rightarrow [0,1]$ where $\sum_{a \in S_q}\rho(a,q)=1$ and $\rho(a,q)=0$ for every $a \notin S_q$. The choice rule describes the DM's behavior. The probability she chooses action $a$ given dataset $q$ is $\rho(a,q)$. We take $\rho$ to be observable.

We adopt the following notational conventions. For a set $B$, $\Delta(B)$ denotes the set of finite-support probability measures on  it. For a set $J \subset N\equiv  \{0,1,\dots,n\}$, $\mathcal{X}_J= \prod_{j\in J} \mathcal{X}_j$,   $x_J$ denotes the event $\{ y\in \mathcal{X}: y_j=x_j \text{ for all }j\in J\}$ for any  $x \in \mathcal{X}_{J'}$ with $J \subseteq J'$, and  $\marg_J p$ denotes the marginal distribution of $p$ on   $\mathcal{X}_J$ when $p \in \Delta(\mathcal{X}_{J'})$ for $J \subseteq J'$.
With slight abuse of notation, we sometimes write $x_j$ instead of $x_{\{j\}}$ and  $x_\emptyset$ for an arbitrary constant random variable.
For $p \in \Delta(\mathcal{X} )$ and disjoint sets $A,B \subset N$, write $X_A \perp_{p} X_B$ if $p(x_{A\cup B})=p(x_{A})p(x_{B})$ for all $x \in \mathcal{X}$.
Finally, for $p_1,p_2\in \Delta(\mathcal{X}_i)$, say that $p_1$ (strictly) FOSD $p_2$ if $p_1((-\infty,y))\leq p_2((-\infty,y))$ for all $y$ (and without equality for some $y$).

\subsection{Causal Graphs}

We model perception of causality  using  a directed acyclic graph (DAG) $R$ over the set $N$.
The DAG $R$ is an acyclic  binary relation, with $iRj$ denoting $(i,j)\in R$ and indicating that   $X_i$ is a direct cause of $X_j$.
Visually, $R$ describes the set of directed edges in a graph with node set $N$.

The following standard terminology will be useful.
The \emph{parents} of $i$, denoted $R(i)$, are the variables that directly cause $X_i$.
The node $j$ is an \emph{ancestor} of $i$, denoted $j \in An_R(i)$, if there are $i_1,\dots,i_m \in N$ so that  $j R i_1 R i_2 R\dots R i_m R i$, and $j$ is a \emph{descendant} of $i$ if $i \in An_R(j)$.
The tuple $(i,j,k)$ is an \emph{$R$-v-collider} if $i Rk$, $j R k$, $j \notr i$, and $i \notr j$.

We follow \cite{Pearl1995Causal} in assuming that if  one intervenes to force $X_0$  to equal  $x_0 $ without observing the realization of any other variables,  the DAG $R$ predicts that the random vector $X$ equals $x\in \mathcal{X}_{-0}$ with probability 
\begin{equation}
q_R(x \mid do(x_0))= \prod_{j = 1}^n q \left( x_j \mid x_{R(j)} \right) 
\label{eq: factorization}
\end{equation}
given the dataset $q$.
The intervention only affects the distribution of variables caused by it, which in turn only affect  variables caused by them. 
Equation (\ref{eq: factorization}) makes  $X_i$ is independent of $X_j$ conditional on $X_{R(i)}$ whenever $j$ is not a descendant of $i$.
In particular, the intervention does not change the distribution of its  parents, even if they are correlated with it.
Given  the DAGs in Figure \ref{fig:intro example DAGs} and a dataset $q$, $q_{R'}(x_{\consequence}|do(x_0))$ adjusts for the potential confounding effect of competitor output, whereas $q_{R} (x_{\consequence}|do(x_0))$ does not.

\subsection{Subjective Causality Representation}
We consider a DM who cares about the realization of only one of the variables, referred to as the  \emph{outcome} and labeled $n^*$.
The other  dimensions correspond to \emph{covariates} that aid the DM in her inferences about the distribution of the outcome.
The DM  has a fixed  subjective causal model  given by the DAG $R$.
The DAG describes her beliefs about the causal relationships between variables but does not  impose any restrictions on the sign or the magnitude of the effects. 
She uses the dataset $q$ and DAG $R$ to predict the outcome distribution resulting  from the intervention ``take action $a$'' for each $a\in S_q$, according to Equation (\ref{eq: factorization}). 
Then, she chooses the action  with  the largest predicted expected utility. 

Our representation restricts attention to a nontrivial and perfect DAG. 
A DAG $R$ is \emph{nontrivial} if $0\in An_R(n^*)$ and \emph{perfect} if there are no $R$-v-colliders.
Neither places any restrictions on  $q \in \mathcal{D}$.
If $R$ is trivial, then  $q_R(x_{n^*}|do(a))$ is  independent of $a$ for every $x\in \mathcal{X}$ and $a \in S_q$.
A perfect DAG implies that the DM does not 
neglect the correlation between two variables used to predict the realization of another variable. We discuss perfection further in Section \ref{sec:remarks}.

\begin{defn}
Choice rule $\rho$ has a \emph{subjective causality representation (SCR)} if there exists a perfect, nontrivial DAG $R$, index $n^* \in N\setminus\{0\}$, and  strictly increasing $u$ so that for any $q \in \mathcal{D}$ and $a \in S_q$, $\rho(a,q)\geq \rho(b,q)$ for all $b\in S_q$  $\iff$ $$U_q(a)=\sum_{x\in \mathcal{X}_{-0}}u(x_{n^*}) q_R(x_{n^*}\mid do(a))\geq \sum_{x\in \mathcal{X}_{-0}}u(x_{n^*}) q_R(x_{n^*}\mid do(b))=U_q(b)$$
for all $b \in S_q$. 
Then, we say that $R$ represents $\rho$ and that $\rho$ has an SCR $(R,u,n^*)$.
\end{defn} 

The DM's most frequent choice (or choices) maximizes expected utility but with a potentially incorrect prediction about the outcome distribution. 
She expects to receive outcome $x_{n^*}$ with probability $q_R(x_{n^*}|do(a))$ if she takes action $a$.

\subsection{Running Example}\label{sec: example}
Throughout, we will illustrate using the following example. The random choice rule represents the frequency with which a management consultant recommends
the different sizes of the workforce based on observational data. The data available to the consultant
is comprised of four variables: the size of the firm's workforce (indexed by $0$), its own output (indexed by $\vari  = 1$), its competitor's output (indexed by $\covariate = 2$), and its profit (indexed by $\consequence = 3$). 
The consultant is hired to maximize
the firm's profit. When choosing the size of workforce, she does not know the realizations of any of the variables, but she has access to a dataset $q$ that contains their joint distribution. 
The vector $(a,x_{\vari },y_{\covariate},z_{\consequence})$ represents a situation in which the size of workforce is given by $a$, own output is $x$, competitor output is $y$ and profit is $z$; $q(a,x_{\vari },y_{\covariate},z_{\consequence})$ is the frequency with which a situation with these characteristics appears in the dataset. 

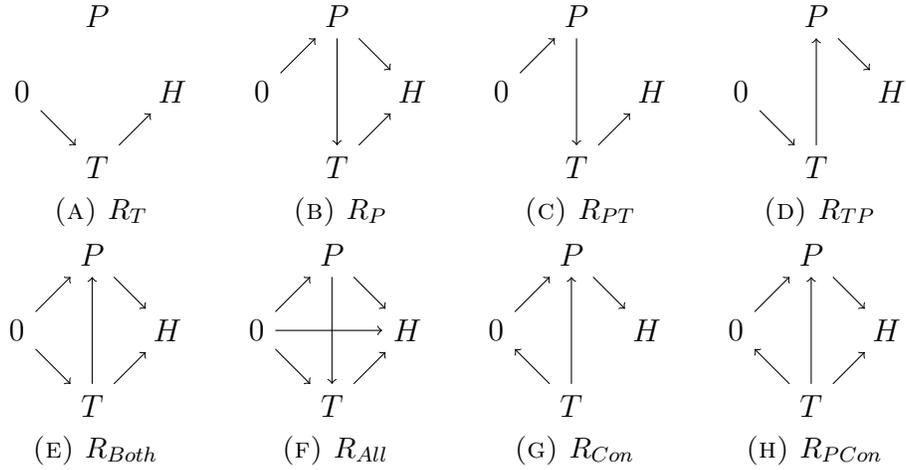
\begin{figure}
\caption{Possible DAGs in the Running Example}
	\begin{minipage}{0.2\textwidth}
	\centering
\begin{tikzpicture}[scale=1]
\node (v0) at (3,1) {$0$};
\node (vE) at (4,2) {$\vari$};
\node (vA) at (4,0) {$\covariate$};
\node (vP) at (5,1) {$\consequence$};
\draw [->] (v0) edge (vA);
\draw [->] (vA) edge (vP);
\end{tikzpicture}
	\vspace{-1.5em}
	\subcaption[first caption.]{$R_{\covariate}$}
	\end{minipage}
	\begin{minipage}{0.2\textwidth}
	\centering
\begin{tikzpicture}[scale=1]
\node (v0) at (3,1) {$0$};
\node (vE) at (4,2) {$\vari$};
\node (vA) at (4,0) {$\covariate$};
\node (vP) at (5,1) {$\consequence$};
\draw [->] (v0) edge (vE);
\draw [->] (vE) edge (vP);
\draw [->] (vE) edge (vA);
\draw [->] (vA) edge (vP);
\end{tikzpicture}
	\vspace{-1.5em}
	\subcaption[first caption.]{$R_{\vari }$}
	\end{minipage}
	\begin{minipage}{0.2\textwidth}
	\centering
\begin{tikzpicture}[scale=1]
\node (v0) at (3,1) {$0$};
\node (vE) at (4,2) {$\vari$};
\node (vA) at (4,0) {$\covariate$};
\node (vP) at (5,1) {$\consequence$};
\draw [->] (v0) edge (vE);
\draw [->] (vA) edge (vP);
\draw [->] (vE) edge (vA);
\end{tikzpicture}
	\vspace{-1.5em}
	\subcaption[first caption.]{$R_{\vari \covariate}$}
	\centering
	\end{minipage}
	\begin{minipage}{0.2\textwidth}
	\centering
\begin{tikzpicture}[scale=1]
\node (v0) at (3,1) {$0$};
\node (vE) at (4,2) {$\vari$};
\node (vA) at (4,0) {$\covariate$};
\node (vP) at (5,1) {$\consequence$};
\draw [->] (v0) edge (vA);
\draw [->] (vE) edge (vP);
\draw [->] (vA) edge (vE);
\end{tikzpicture}
	\vspace{-1.5em}
	\subcaption[first caption.]{$R_{\covariate\vari }$}
	\end{minipage}
	
	\begin{minipage}{0.2\textwidth}
	\centering
\begin{tikzpicture}[scale=1]
\node (v0) at (3,1) {$0$};
\node (vE) at (4,2) {$\vari$};
\node (vA) at (4,0) {$\covariate$};
\node (vP) at (5,1) {$\consequence$};
\draw [->] (v0) edge (vE);
\draw [->] (v0) edge (vA);
\draw [->] (vA) edge (vP);
\draw [->] (vA) edge (vE);
\draw [->] (vE) edge (vP);
\end{tikzpicture}
	\vspace{-1.5em}
	\subcaption[first caption.]{$R_{Both}$}
	\end{minipage}
	\begin{minipage}{0.2\textwidth}
	\centering
\begin{tikzpicture}[scale=1]
\node (v0) at (3,1) {$0$};
\node (vE) at (4,2) {$\vari$};
\node (vA) at (4,0) {$\covariate$};
\node (vP) at (5,1) {$\consequence$};
\draw [->] (v0) edge (vE);
\draw [->] (v0) edge (vA);
\draw [->] (v0) edge (vP);
\draw [->] (vA) edge (vP);
\draw [->] (vE) edge (vA);
\draw [->] (vE) edge (vP);
\end{tikzpicture}
	\vspace{-1.5em}
	\subcaption[first caption.]{$R_{All}$}
	\end{minipage}
	\begin{minipage}{0.2\textwidth}
	\centering
\begin{tikzpicture}[scale=1]
\node (v0) at (3,1) {$0$};
\node (vE) at (4,2) {$\vari$};
\node (vA) at (4,0) {$\covariate$};
\node (vP) at (5,1) {$\consequence$};
\draw [->] (v0) edge (vE);
\draw [->] (vA) edge (v0);
\draw [->] (vE) edge (vP);
\draw [->] (vA) edge (vE);
\end{tikzpicture}
	\vspace{-1.5em}
	\subcaption[first caption.]{$R_{Con}$}
	\end{minipage}
	\begin{minipage}{0.2\textwidth}
	\centering
\begin{tikzpicture}[scale=1]
\node (v0) at (3,1) {$0$};
\node (vE) at (4,2) {$\vari$};
\node (vA) at (4,0) {$\covariate$};
\node (vP) at (5,1) {$\consequence$};
\draw [->] (v0) edge (vE);
\draw [->] (vA) edge (v0);
\draw [->] (vE) edge (vP);
\draw [->] (vA) edge (vE);
\draw [->] (vA) edge (vP);
\end{tikzpicture}
	\vspace{-1.5em}
	\subcaption[first caption.]{$R_{\covariate Con}$}
	\end{minipage}
\label{fig:example DAGs}
\end{figure}

Figure \ref{fig:example DAGs} gives some possible DAGs, each representing a different theory of causation.%
\footnote{Given an application some DAGs might be pre-emptively ruled out as unreasonable, such as $R_{\covariate }$ in Figure \ref{fig:example DAGs}. Our approach does not rely on pre-existing knowledge but can accommodate it.
}
A consultant whose causal model is 
represented by $R_{\vari }$ believes that workforce size directly influences only own output, that competitor output is directly influenced by own output, and that both outputs directly influence profits. Roughly, the consultant believes that the firm is a Stackelberg leader.
By contrast, one whose causal model is represented by $R_{Both}$ believes that  workforce size directly influences both own and competitor output (e.g., as a signal of production capacity), which in turn influence profit.
The DAG $R_{\covariate Con}$ differs from $R_{Both}$ only in that competitor output is believed to have affected the choice of workforce size in the data rather than vice versa. This may occur due to a disagreement about the speed at which the competitor can change their output. 

While both $R_{\covariate Con}$ and $R_{Both}$ explain any dataset equally well, they may lead to opposite predictions about which action is better.%
\footnote{Any data consistent with $R_{\covariate Con}$ is consistent with $R_{Both}$ by Thoerem 1 of \cite{Pearl1991}.}
Consider the dataset $q$ described by Table \ref{tab: example data}, letting $1$ represent the larger realization of all variables and  $0$ represent the smaller realization. Moreover, assume that own output is high or low (almost always) when workforce is big or tiny, i.e.    $q(0_{\vari}|y_{\covariate },a)\approx 1$ and  $q(1_{\vari}|y_{\covariate },b)\approx 1$ for each $y \in \{0,1\}$.
Now, if the DM  is represented by $R_{Both}$, then she takes action $a$ since 
\begin{align*}
q_{R_{Both}}(1_{\consequence}|do(a))= &
\sum_{y=0,1} q(y_{\covariate }|a)
\left[\sum_{x=0,1}q(x_{\vari }|y_{\covariate },a)  q(1_{\consequence}|x_{\vari },y_{\covariate }) \right] \\
\approx & 0.2*(1* 0.7 + 0* 0.3) + 0.8 * (1*0.3 + 0* 0.7) = 0.38\\ 
< & 0.52 = 0.8*0.6 + 0.2 * 0.2 \approx q_{R_{Both}}(1_{\consequence}|do(b)).
\end{align*} 
If she  is instead represented by $R_{\covariate Con}$, then she takes action $b$ since
\begin{align*}
q_{R_{\covariate Con}}(1_{\consequence}|do(a))= &
\sum_{y=0,1}q(y_{\covariate } )
\left[ \sum_{x=0,1}q(x_{\vari }|y_{\covariate },a ) q(1_{\consequence}|y_{\covariate },x_{\vari })  
\right] \\
\approx & \frac12*(1* 0.7 + 0*0.3) + \frac12*(1*0.3 + 0*0.7) = 0.5 \\
> & 0.4 = \frac12 *0.6 + \frac12 * 0.2  \approx  q_{R_{\covariate Con}}(1_{\consequence}|do(b)).
\end{align*}
The expression $q_{R_{\covariate Con}}(1_{\consequence}|do(a))$ treats competitor output as a confounder affecting the relationship between workforce size and own output. In contrast, $q_{R_{Both}}(1_{\consequence}|do(a))$ treats it as an endogenous consequence of workforce size. The consultant's beliefs about causality map the same correlations into different predictions.

\subsection{Remarks on the Model}
\label{sec:remarks}
The agent is dogmatic about her model $R$, in that she does not consider switching  models because of her dataset. This conforms with evidence from psychology that, as summarized by \citet[p. 107]{Sloman2005}, ``causal explanations quickly become independent of the data from which they are derived.'' We view our setting as one in which the DM has already formulated her model, perhaps from past experience or earlier data, and does not deviate from it on the basis of the present dataset.

We focus on perfect DAGs for philosophical and psychological reasons. Perfection limits how wrong the DM's predictions can be, as it ensures that the DM correctly predicts the marginal distribution of individual variables.  Psychology experiments \citep{lombrozo2007simplicity,pacer2017ockham} indicate  a preference for stories with fewer unexplained variables, and a perfect DAG has exactly one unexplained variable.  Moreover, perfect DAGs are widely applied. Most of  the literature following \cite{Spiegler2016} uses perfect subjective DAGs. The first step of many algorithms using DAGs is to make the DAG perfect if it is not already (e.g. Chapter 6 of \cite{cowell1999}). Natural procedures, such as  extrapolating from multiple datasets by filling in missing correlations to maximize  entropy,  can be represented as coming from a perfect DAG \citep{Spiegler2017}.  

The ideas behind  our identification results apply to imperfect DAGs as well.  V-colliders introduce complications similar to confounding variables (defined formally in the next section). Confounders can be treated as v-colliders with the action. However, any v-collider ancestral to the outcome  affects the DM's prediction, even if it is exogenous and far removed from the variables that she thinks her action can affect. While tedious, these can be identified from the DM's behavior in a similar manner as we do with confounders. 

Our identification results only require that the DM's most frequent choices maximize perceived expected utility. This allows application to several different choice procedures. Of particular interest are the following three. First, $\rho(\cdot, q)$ equals the uniform distribution on $\arg\max_{S_q} U_q(\cdot)$. This nests deterministic choice with $c(S_q)=\supp \rho(\cdot, q)$.   Second, $\rho(a,q)$ equals $$\Pr(\{\varepsilon:U_q(a)+ \varepsilon_a\geq U_q(b)+\varepsilon_b\ \forall b\in S_q\})$$ where $\{\varepsilon_i\}_{i\in S_q}$ are independent and identically distributed. This includes Logit and Probit  for particular distributions of $\varepsilon_i$. Our endogenous dataset results focus on Logit so that all actions are chosen with positive probability. Third, $\rho(a,q)>\frac12$ if and only if $U_q(a)>U_q(b)$ when $X_0=\{a,b\}$.

We assume that that the agent cares about influencing a single random variable.
Our results immediately extend to multiple outcomes provided that they are all directly causally related to one another.  Otherwise, the paths between the outcomes may also matter. We leave this as an open question for future work.


\section{Revealing the Subjective Causal Model}
\label{sec:ID}
When the DM is represented by the DAG $R$, three features of it affect her predictions: its confounders, its minimal active paths (MAPs) from the action to the outcome, and its MAPs from a confounder to the outcome. 

A confounder is a variable that the DM believes can influence both her action and another variable, and thereby affect the relationship between them. The node $i^*$ is an \emph{$R$-confounder} if $i^* R0$, $0Rl$ and $i^*Rl$ for some $l\in An_R(n^*)\cup\{n^*\}$.
Competitor output is a $R_{\covariate Con}$-confounder  but not an $R_{Both}$-confounder.

A MAP  is a (directed) path in the DAG that cannot be made shorter by omitting variables.  Paths in the DAG captures a chain  of causal reasoning, and so MAPs capture the simplest ways in which one variable can affect another.
Formally, the finite sequence $(i_0,\dots,i_m)$   is an \emph{$R$-MAP from $i_0$ to $i_m$} if $i_j R i_{k}$ if and only if $k=j+1$ and $i_j$ is an $R$-confounder or $0$ only if $j=0$.
The path $(0,\vari,\consequence)$ is both an $R_{Both}$-MAP and an $R_{\covariate Con}$-MAP,  $(0,\covariate ,\consequence)$ is an $R_{Both}$-MAP but not an $R_{\covariate Con}$-MAP,
and $(0,\covariate ,\vari,\consequence)$ is neither an $R_{Both}$-MAP nor an $R_{\covariate Con}$-MAP.

\begin{thm}
\label{thm: RMAP characterize DAG}
	Let $\rho$ have a SCR $(R,u,n^*)$ and $R'$ be a perfect DAG.
	Then, $\rho$ has an SCR $(R',u,n'{}^*)$ if and only if $n^*=n'{}^*$, $R$ and $R'$ have the same confounders, and $R$ and $R'$ have the same MAPs from the action or a confounder to $n^*$.
\end{thm}

 Theorem \ref{thm: RMAP characterize DAG} shows that  the identity of the outcome, the confounding variables, and the simplest causal paths determine her predictions.
 Paths from her action to the outcome represent the direct and indirect ways that she thinks her action could affect her payoff.
DMs who disagree on the simplest paths make different choices.
The DM does not think that her action affects confounders, and so adjusts for them when forming her predictions. Different confounders, or paths from confounders to the outcome, lead to different adjustments and thus different behavior.

Because a MAP cannot be made shorter, it represents one of the simplest causal chains. The result thus provides a version of Occam's razor:
the simplest ways in which the DM thinks her action (or a confounder) can affect the outcome determine her choices.
For instance, a DM represented by $R_{\vari}$ behaves identically to one represented by $0\rightarrow \vari \rightarrow \consequence$ even though $R_{\vari}$ contains the additional path $(0,\vari,\covariate ,\consequence)$.
The latter path captures a belief that the effect of own output on profits is partly due to its effect on competitor output.
However, this does not affect her estimate of own output's total effect on profits. Because both DMs think that their action can only influence own output, the two make the same predictions (and hence the same choices).
Mathematically, the effect of competitor output on profits integrates out through the identity $\sum_{x_{\vari} \in \mathcal{X}_{\vari}} q(z_{\consequence}|x_{\vari},y_{\covariate })q(x_{\vari}|y_{\covariate })=q(z_{\consequence}|y_{\covariate })$.

\begin{rem}
If $\rho$ has SCRs $(R,u,n^*)$ and $(R',u',n^*)$, then there exists $\alpha>0$ and $\beta$ so that $u(z)=\alpha u'(z)+\beta$ for every $z \in \mathcal{X}_{n^*}$. The exact form of the converse depends on how much structure one puts on the random choice rule. For instance, the converse holds if $\rho(\cdot,q)$ is uniformly distributed on $\arg\max_{S_q} U_q (\cdot)$ and holds with the extra requirement that $\alpha=1$ if $\rho(\cdot,q)$ is Logit or Probit.
\end{rem}

\subsection{Implications}
Theorem \ref{thm: RMAP characterize DAG} has two immediate implications that allow us to simplify the set of DAGs considered.
First, only relationships between variables that appear in at least one $R$-MAP affect the DM's behavior.
In particular, any variables that she thinks are caused by the outcome are inconsequential for her behavior.
Second, the simplest causal chains of causality determine all other relevant causal relationships.
While there may be edges not in an $R$-MAP, their existence and direction can either be determined from the $R$-MAPs or are immaterial to the DM's choices.

The  economic content of Theorem \ref{thm: RMAP characterize DAG} is that identifying the $R$-MAPs, confounders, and the outcome suffices to pin down the DM's inferences.
Any two DMs whose DAGs differ on at least one MAP evaluate some actions differently.
In fact, given ``almost any'' dataset where a subset of variables that includes $n^*$ and $0$ is independent of the others, the utility difference between any pair of actions is uniquely determined by the MAP (or MAPs) inside that subset.
\begin{proposition}\label{prop: MAP nongeneric}
Suppose that $(i_0,\dots,i_m)$ is an $R$-MAP from $0$ to $n^*$ but not an $R'$-MAP from $0$ to $n^*$ and that $I =\{i_0,\dots,i_m\}$.
For every non-constant $u$ and almost every $q \in \mathcal{D}$ so that  $X_I \perp_q X_{I^c}$ and $S_q=\{a,b\}$,   $\sum_x u(x) [q_R(x_{n^*}|a)-q_R(x_{n^*}|b)] \neq \sum_x u(x)[q_{R'}(x_{n^*}|a)-q_{R'}(x_{n^*}|b)]$.
\end{proposition}

The result shows that a DM with model $R$ evaluates actions differently than with model $R'$ whenever $(i_0,\dots, i_m)$ is not an $R'$-MAP. This holds regardless of how $R'$ relates the variables in  $I=\{i_0,\dots,i_m\}$. They may contain no $R'$-MAP  or an $R'$-MAP distinct from $(i_0,\dots, i_m)$.
When we know the DM's tastes and how the DM randomizes, this  makes identification of DM's MAPs relatively easy. The result is particularly applicable in a laboratory setting where preferences can be induced.
If the DM's randomness comes from logit or probit errors, then the frequency with which she chooses action $a$ reveals the utility difference. 
In that case, one can calculate $q_{R'}(x_{n^*}|c)$ for each $x$, $c\in S_q$, and $R'$, and then check whether the DM chooses consistently with that calculation.

The independence needed for this result, as well as Lemmas \ref{lem: I contains AP iff block} and \ref{lem: adjacent iff rmap} below, can be created from any dataset $q$ by providing additional information about the outcomes of certain variables. In an experiment, the experimenter can ask subjects ``what would you choose if the random vector $X_A$ were equal to $x_A$ regardless of your choice?'' Assuming the DM believes this information, the resulting dataset $q'=q(\cdot|x_A)$ satisfies $X_A \perp_{q'} X_{A^c}$. Similar variation would result from policy interventions or court decisions that require bureaucrats to disregard certain information, such as race, gender, or income, when making their decision. Assuming they follow these instructions, they should act as if those variables are independent of all others. An experimenter can also generate the any of the independences needed for Proposition \ref{prop: MAP nongeneric} or Lemma \ref{lem: I contains AP iff block} from a single dataset by splitting it into multiple partial datasets provided that the DM extrapolates from these partial datasets as in \cite{Spiegler2017}.

\subsection{Identification}

We now show how one can identify the $R$-MAPs from behavior without any pre-existing information on the DM's preferences or how she randomizes. 
Observe that the DM infers that the parents of a variable have no causal effect on it when it is independent of them.
We can tease out her DAG from her behavior when she forms her predictions using a dataset where certain sets of variables are independent of all others.

We first identify the sets of variables that intersect every $R$-MAP. To state the lemma, let  $\bar{x} = \left(\max \mathcal{X}_i \right)_{i=1}^n$ and  $\underline{x} = \left(\min \mathcal{X}_i \right)_{i=1}^n$. 

\begin{lemma}
\label{lem: I contains AP iff block}
If $\rho$ has an SCR $(R,u,n^*)$ and $0\notin A$, then the following are equivalent:
\begin{enumerate}[(i)]
\item every $R$-MAP from $0$ to $n^*$ intersects $A$ 
\item  for any dataset $q$, $\rho(a,q)=\rho(b,q)$  for all $a,b \in S_q$ whenever $X_A \perp_q X_{A^c}$
\end{enumerate}
Moreover, if $R$ is uncounfounded and $q$ is a dataset having $S_{q}=\{a,b\}$,  $X_A \perp_{q} X_{A^c}$, and  ${q}(\bar{x}_{A^c}|a),{q}(\underline{x}_{A^c}|b)>   \left((2^{1/n} -1){q}(a)+1 \right)^{-1}$, then either is equivalent  to
\begin{enumerate}[(i)]
\setcounter{enumi}{2}
\item $\rho(a,{q})=\rho(b,{q})$
\end{enumerate}   
\end{lemma}

If Lemma \ref{lem: I contains AP iff block}.(ii) holds for $A$, then the DM is indifferent between her available actions whenever $X_A$ is independent of the other variables in her dataset. 
This reveals that every $R$-MAP intersects $A$.
To illustrate why, consider a consultant who is equally likely to recommend every workforce size whenever own output is independent of the other variables in the data.
When large profits and high competitor output occur with very high probability for big workforce and with very low probability for tiny workforce, she is indifferent between the two actions only if she does not think the positive correlations are causal. This requires that every causal chain in her DAG includes own output. 

Based on Lemma \ref{lem: I contains AP iff block}, we introduce the following definitions.
\begin{defn}
\label{def: separator}
A set $A$ of variables is a \emph{$\rho$-separator} if either $0 \in A$ or Lemma \ref{lem: I contains AP iff block}.(ii)  holds for $A$, and   $\mathcal{A}^\rho$ is the set of $\subseteq$-minimal $\rho$-separators.
A \emph{selection from the $\rho$-separators} is a $\subseteq$-minimal set $I \subseteq N$ so that $I \cap A \neq \emptyset$  for all $A \in \mathcal{A}^\rho$. 
\end{defn}
By  Lemma \ref{lem: I contains AP iff block}, $A$ is a $\rho$-separator if it intersects every $R$-MAP from $0$ to $n^*$.
An $R$-MAP must contain at least one index from every $A \in \mathcal{A}^\rho$ but cannot contain too many from any of them.
This suggests, and Lemma \ref{lem: adjacent iff rmap} proves, that a selection from the separators $I$ contains the variables in an $R$-MAP. 

We also say that a dataset $q$ has \emph{$\{i ,j \}$-MLRP} if  $S_q = \{a,b\}$ and $\frac{q(x_{i }|y_{j })}{q(x'_{i }|y_{j })}\geq \frac{q(x_{i }|y'_{j })}{q(x'_{i }|y'_{j })} $  when $x >x'$ and $y>y'$, without equality for at least one  $(x,x',y,y')$ and where ``$a>b$'' is taken to be true. This implies that a larger value of $X_i$ increases the likelihood of higher values of $X_j$ and vice versa.
\begin{defn}
\label{def: rho adjacent}
The variable \emph{$i $ is  $\rho$-adjacent to $j  \neq i$}  if there is a selection from the $\rho$-separators $I$ with $i,j \in I$ and $\rho(a,q)=\rho(b,q)$ for some $q\in \mathcal{D}$ for which
$$X_I \perp_{q} X_{I^c} \ \& \ X_{i} \perp_{q} X_{j}$$ 
and that   has $\{i',j'\}$-MLRP for all distinct $i',j' \in I$ s.t. $\{i',j'\} \neq \{i , j\},\{k,0\}$ for some $k \in \{i,j\}$ that is not $\rho$-adjacent to $0$.
\end{defn}

Note that $0$ is never $\rho$-adjacent to itself, so we can use $k=0$ to determine whether or not $0$ is $\rho$-adjacent to $j$ for any $j$.
We can now reveal the $R$-MAPs.

\begin{lemma}
\label{lem: adjacent iff rmap}
If $\rho$ has an SCR $(R,u,n^*)$, then for any $I=\{i_0,\dots,i_m\}$ with $i_0=0$, the following are equivalent:
\begin{enumerate}[(i)]

\item $(i_0,\dots,i_m)$ is an $R$-MAP from $0$ to $n^*$ 

\item   $I$ is a selection from the $\rho$-separators, and for any  $j=1,\dots,m$ and $q \in \mathcal{D}$, $X_I \perp_q X_{I^c}$ and $X_{i_{j-1}} \perp_q X_{i_{j}}$ imply that $\rho(a,q)=\rho(b,q)$ for all $a,b\in S_q$  

 
\item $I$ is a selection from the $\rho$-separators and for each $j=1,\dots,m-1$, $i_j$ is $\rho$-adjacent to $i_{k}$ if and only if $k \in \{j-1,j+1\}$.
\end{enumerate}
\end{lemma}

The result first shows that any selection from the $\rho$-separators $I$ contains the variables in an $R$-MAP. Therefore when $j\in I$, the DM believes that the variable $X_j$ is caused by exactly one other variable in $I$, say $X_i$. Since $X_I \perp_q X_{I^c}$, her prediction about the value of $X_j$ that results from taking an action is a function only  of her prediction about $X_i$ and the relationship between   $X_i$ and $X_j$. 
If  $X_{i} \perp_q X_{j}$, then she estimates its causal effect to be zero.
If the two are positively correlated (as when the dataset satisfies $\{i,j\}$-MLRP), then she infers that a higher value of $X_i$ causes a higher value of $X_j$. Therefore, she expresses indifference between $a$ and $b$ for every dataset $q$ as in Lemma \ref{lem: adjacent iff rmap}.(ii) if and only if she believes that $i$ causes $j$ (or vice versa).

Note that the lemma  also identifies $n^*$ as the element of $\{1,\dots,n\}$ adjacent to exactly one variable in every selection. 
%

Combined, the lemmas allow us to determine if $(i_0,\dots,i_m)$ is an $R$-MAP from $0$ to $n^*$ from a finite number of observations. 
First, check whether $I=\{i_0,\dots,i_m\}$ is a selection from the separators using Lemma \ref{lem: I contains AP iff block}. If so, then   pick datasets $q_k$ for $k=1,\dots,m$ satisfying the conditions of Definition \ref{def: rho adjacent} for $I$ and $\{i,j\}=\{i_k,i_{k-1}\}$.
If $\rho(a,q_k)=\rho(b,q_k)$ for all $k$, then the equivalence of Lemma \ref{lem: adjacent iff rmap}.(i) and Lemma \ref{lem: adjacent iff rmap}.(iii) gives that $(i_0,\dots,i_m)$ is an $R$-MAP. Otherwise, $\rho(a,q_k) \neq \rho(b,q_k)$ for some $k$, so the equivalence of Lemma \ref{lem: adjacent iff rmap}.(i) and Lemma \ref{lem: adjacent iff rmap}.(ii) gives that $(i_0,\dots,i_m)$ is not an $R$-MAP.

When the DM believes there are confounding variables, all representations must also have the same confounders and the same minimal causal chains from each confounder to the outcome.
To identify the presence of confounders in the DM's causal model, we observe that in their absence, only the distribution over variables involved in a minimal causal chain from $0$ to $n^*$ affects the DM's choices. 
If the action is unconditionally independent of all variables involved in any minimal causal chain from $0$ to $n^*$ and the DM is not indifferent between all of her actions, then her causal model contains at least one confounder.%
\footnote{Lemma \ref{lem: adjacent iff rmap} still applies, so we can identify the $R$-MAPs from $0$ to $n^*$ using it.}
More precisely, if $i^*$ is not part of any selection from the $\rho$-separators, then $i^*$ is a \emph{revealed confounder} if the DM expresses a strict preference between two actions for some dataset in which every variable except $i^*$ is independent of her action.
We formalize this definition and show that they must be confounders in the DM's DAG in Appendix \ref{sec: Proof Necessity ID}.
We then extend the procedure above to identify the MAPs from each revealed confounder to the outcome.

\section{Endogenous Datasets}
\label{sec:Foundation}
Up to now, we have focused on a situation where the data from which the DM derives her predictions is exogenously given. 
We now turn to the case where the dataset utilized by the DM is generated by her own behavior as in \cite{Spiegler2016}. This section introduces the Endogenous Subjective Causality Representation (Endogenous SCR) and shows that it can accommodate various biases from the psychology literature. The next section axiomatically characterizes the random choice rules with an Endogenous SCR.

To model this, we take as given the (full-support) distribution over $\mathcal{X}_{-0}$ that would result from taking each action $a \in \mathcal{X}_0$. This distribution should be interpreted as known to us but not the DM, who leverages her causal model to learn it.
We assume that  $\mathcal{X}_0$ is rich enough so that there is at least one action corresponding to every full-support lottery over $\mathcal{X}_{-0}$, and to economize on notation, we write $a(x)$ for the chance that $x \in \mathcal{X}_{-0}$ results when action $a$ is taken.
	
Let $\mathcal{S}$ be the set of non-empty, finite subsets of $\mathcal{X}_0$.
Then, a random choice rule $\rho: \mathcal{X}_0 \times \mathcal{S} \rightarrow [0,1]$ where $\sum_{a \in S}\rho(a,S)=1$ and $\rho(a,S)=0$ for every $a \notin S$ maps each feasible set $S \in \mathcal{S}$ to the probability the DM chooses each action $a \in \mathcal{X}_0$. 
The DM uses the dataset $\rho^S$ when facing menu $S$  where
\[ \rho^S(a,y)= \rho(a,S) a(y)\text{ for all }a\in \mathcal{X}_0\text{ and } y\in \mathcal{X}_{-0}.\]
The dataset $\rho^S$ combines the lottery  that results from taking action $a$ with the frequency that $a$ is chosen, $\rho(a,S)$.
Note that $\rho(\cdot,S)$ is a distribution over actions whereas $\rho^S$ is a distribution over $\mathcal{X}$.

\begin{defn}
Random choice rule $\rho$ has an \emph{Endogenous SCR} if there exists a perfect, nontrivial, unconfounded DAG $R$, an index $n^*$, and a strictly increasing  $u:\mathcal{X}_{n^*} \rightarrow \mathbb{R}$ so that
\[ \rho(a,S)=\frac{\exp \left(  \sum_{x\in \mathcal{X}_{-0}}u(x_{n^*}) \rho^S_{R}(x |do(a)) \right) }{\sum_{a' \in S} \exp \left(\sum_{x\in \mathcal{X}_{-0}}u(x_{n^*}) \rho^S_{R}(x |do(a')) \right)}\]
for every $S \in \mathcal{S}$ and $a \in S$. 
\end{defn}

	As before, the DM uses her model $R$ to form predictions about the outcomes of each potential action.
	However, the dataset $\rho^S$ that she uses to estimate each of the causal effects varies with how often she chooses each of the actions.
	Note that 
	$\rho(a,S)=\Pr(\{\varepsilon:U_{\rho^S}(a)+ \varepsilon_a\geq U_{\rho^S}(b)+\varepsilon_b\text{ for all }b\in S\})$
	when $\{\varepsilon_b\}_{b\in S}$ are distributed independently and extreme value.
	Moreover, the representation is a personal equilibrium \citep{KR06}: the DM maximizes expected utility given her beliefs, which in turn depend on her choices.
It is easy to show that an equilibrium exists for any $S \in \mathcal{S}$ using Brouwer's fixed point theorem.
For menus with more than one equilibrium, we place no restrictions on which is selected.
\begin{rem}
Random choice plays two roles. First, $\rho(a,S)>0$ for all $a \in S$, so predictions about every action are well-defined. Second, \cite{Spiegler2016} shows that ``pure'' equilibria may not exist, even after defining beliefs about unchosen actions.
\end{rem}

\begin{rem}
We adopt the exponential function for concreteness and applicability.
Our results adapt to any other strictly increasing and positive function.
\end{rem}

\begin{rem}
A confounded DAG would represent a DM who believes that a variable whose realization she does not observe affects her choices.
\end{rem}

The DM's behavior may endogenously create correlations that she misinterprets as causation.
Fundamentally, the DM neglects the effect of her choices on the variables that she thinks are not directly affect by her action.%
\footnote{\cite{esponda2018endogenous} experimentally document selection neglect, and \cite{denrell2020sampling} provides a recent survey of evidence for it in managers.}
This leads to two key technical challenges.
First, the DM  may violate \emph{regularity}, a property necessary for representation by a random utility model (RUM). 
Second, her behavior may be \emph{self-confirming}:
because she chooses $a$ frequently, she thinks it is better than $b$, but she would reverse her ranking if she chose $b$ more frequently.
These features allow the model to accommodate a number of biases documented in the psychology literature, including violations of independence of irrelevant alternatives, illusion of control, status-quo bias, and congruence bias.
We illustrate these challenges using a consultant in the running example whose behavior has an Endogenous SCR $(R_{\vari },u,\consequence)$ and where $\mathcal{X}_i=\{0,1\}$ for $i=\vari ,\covariate ,\consequence$.

\subsection{Regularity violation}
\label{sec: reg violation}
An Endogenous  SCR may violate regularity, the requirement that  $\rho(a,S) \geq \rho(a,S')$ whenever  $a \in S \subseteq S'$. 
Consequently, an Endogenous SCR is not a RUM.
By contrast, many other models of ``irrational'' behavior are RUMs, such as the model of limited attention due to \cite{manzini2014stochastic}. 

To illustrate why regularity may be violated, consider three sizes of the workforce, $a,b,c$, that are equally likely to lead to high profit.
The firm's own output and profit are independent when workforce size is $a$, 
positively correlated when it is $b$, and negatively correlated when it is $c$.
When the consultant can only recommend $a$ or $b$, own output is necessarily positively correlated with profit.
As she mistakes the correlation for causation, this makes the workforce size that is more likely to lead to high  output more attractive. 
However, when she chooses between all three workforce sizes, the recommendations of $c$ might cancel out or even reverse the perceived positive effect of own output on profit.
When this effect is strong enough, she may reverse her recommendation and violate regularity.
Formally, suppose that $a (1_{\vari },1_{\consequence})=a (1_{\vari },0_{\consequence})=\frac12$, $b (0_{\vari },0_{\consequence})=b (1_{\vari },1_{\consequence})=\frac12$, $c (1_{\vari },0_{\consequence})=c (0_{\vari },1_{\consequence})=\frac12$, $u(0)=0$, and $u(1)=6$.%
\footnote{The distribution of competitor output does not affect behavior, so we leave it unspecified.}
One can verify that for $S=\{a,b\} \subset S'=\{a,b,c\}$, we have $\rho(b,S)<\frac13=\rho(b,S')$.%
\footnote{See Appendix \ref{ex: regularity} for the derivation.}$\ \!^{,}$%
\footnote{We note that $S,S' \notin \mathcal{S}$, so this example, and that in the next subsection, are technically outside our domain. At the cost of complicating the algebra and obscuring the logic, they can be made consistent with our assumptions by replacing each action $d$ with $d'=(1-\epsilon)d+\epsilon e$ where $\epsilon>0$ is small enough and $e(y)=\frac18$ for each $y\in \{0,1\}^3$.}

While the violations of regularity allow the model to accommodate phenomena like the decoy effect, the failure stems from faulty reasoning.
The above consultant overestimates her ability to control events, or exhibits illusion of control \citep{Langer1975illusion}. 
Her choices do not affect profit, yet she would be willing to pay a premium to choose one workforce size over another. 
Moreover, she also exhibits ``patternicity'' \citep{Shermer1998} in that she perceives a pattern that does not exist, namely that taking action $a$ in menu $\{a,b\}$ leads to a higher expected profit.

\subsection{Self-confirming choices}
\label{sec:self-confirming}
The outcome that the DM expects to get from an action may depend on how frequently she chooses it. She may predict that one action is better than another only if she chooses it sufficiently frequently. This can lead to multiple personal equilibria.

Consider the following setting. Workforce size $b$ very often leads to both high own output and high profit. Workforce size $a$ can lead to either high or low profit as well as either low or high own output, and low own output more often leads to high profit.
When the consultant rarely recommends workforce size $b$, the low own output is mistakenly believed to increase the probability of high profit. Consequently, recommending $b$, which raises the probability of high own output, seems like a bad idea.
Symmetrically, when she usually recommends $b$, she predicts that $a$ decreases the probability of high profit, because then high own output is positively correlated with profit.
Formally, let $b(1_{\vari },1_{\consequence})=1$, $a(0_{\vari },1_{\consequence})=\frac12 q$, $a(0_{\vari },0_{\consequence})=\frac12 (1-q)$, and $a(1_{\vari },0_{\consequence})=\frac12$ where $q \in (0,1)$.
When  $u(0)=0$ and $u(1)=30$, one can verify that 
$\rho(a,\{a,b\}) \approx .34$, $\rho(a,\{a,b\}) \approx 0.02$, and $\rho(a,\{a,b\}) \approx 0.99$ are all equilibria when  $q=\frac34$.%
\footnote{See Appendix \ref{ex: selfconfirming} for the derivation.}

If we interpret the choice rule as the steady state of a learning process, then the above DM exhibits two biases. First, she exhibits status quo bias \citep{SamuelsonZeckhauser88}, a tendency toward ``maintaining one's current or previous decision.''  Second, she exhibits congruence bias \citep{Wason1960Congruence} by failing to test the alternative hypothesis that the other workforce size is better.
 
\section{Behavioral Foundations}
We now turn to the behavioral regularities that characterize the random choice rules with an Endogenous SCR.
We first identify a candidate causal model for the DM.
The axioms relate the predictions implied by that DAG  to the DM's choices.
Section \ref{sec: basic axioms} presents some familiar axioms. 
Section \ref{sec: special axioms} first illustrates our approach by characterizing the recommendation of a consultant in the running example  with model $R_{\vari } $.
Section \ref{sec: general axioms} presents the axiomatization for the general case.   

\subsection{Basic axioms}
\label{sec: basic axioms}
The first two axioms are standard, and we present them with minimal discussion. The first requires that the DM chooses every  action with positive probability.

\begin{ax}[Full-support]
\label{ax: positivity}
\label{ax: first SCR}
\label{ax: nontrivial}
For any $S \in \mathcal{S}$ and $a\in S$, $\rho(a,S)>0$.
\end{ax}

The second limits the perceived difference between any two options. 

\begin{ax}[Bounded Misperception]
\label{ax: bounded}
The quantity $\sup_{S, a,b\in S} \frac{\rho(a,S)}{\rho(b,S)}$ is finite.
\end{ax}

The relative frequency with which the DM takes two actions, their \emph{Luce ratio}, indicates the strength of her preference.
Since the set of outcomes is finite, there is a best and worst outcome. 
These provide a natural limit to how much she prefers one action to another, which bounds the Luce ratio. 
The axiom thus bounds the size of the mistakes that the DM can make.
 
\subsection{Foundations - special case}
\label{sec: special axioms}
This subsection specializes the remaining axioms to a consultant represented by $R_{\vari }$. The first requires that $R_{\vari }$ is a candidate for her DAG.
\begin{ax*}[Consistent Revealed Causes$^*$, CRC*]
The set of minimal separators $\mathcal{A}^{\rho}$ equals $\{\{0\},\{\vari \},\{\consequence\}\}$.
If $X_{0} \perp_{\rho^{\{a,b\}}} X_{\vari }$ or $X_{\vari } \perp_{\rho^{\{a,b\}}} X_{\consequence}$, then $\rho(a,\{a,b\})=\frac{1}{2}$, and if $\rho^{\{a,b\}}$ has $\{0,\vari \}$- and $\{\vari,\consequence \}$-MLRP, then $\rho(a,\{a,b\})>\frac12$.
\end{ax*}

Under CRC*, the only MAP is $(0,\vari ,\consequence)$ by Lemma \ref{lem: adjacent iff rmap}. Moreover, since $\consequence$ is furthest from $0$ in that $R$-MAP,  the DM cares only about profit. Therefore, if the consultant has an Endogenous SCR, then she has one where her DAG is $R_{\vari}$.

Second, the consultant treats own output as a sufficient statistic for profit.
\begin{ax*}[Indifferent If Identical Immediate Implications$^*$, I5*]For $a,b \in S \in \mathcal{S}$, if $\marg_{\vari } a= \marg_{\vari } b$, then $\rho(a,S)=\rho(b,S)$.
\end{ax*}
According to $R_{\vari }$, any workforce sizes that lead to the same distribution of own output also have the same probability of large profit.
Therefore, if two workforce sizes lead to the same marginal distribution over own output, then the consultant is equally likely to choose either, regardless of any other differences between them. 

Third, she chooses similarly  when she makes similar predictions. 
\begin{ax*}[Luce’s Choice Axiom Given Inference$^*$, LCI*]
	For    any $S,S_1,S_2, \dots \in \mathcal{S}$ with $a,b \in S_m \cap S$ for each $m$:\\
	if 
$
\rho^{S_m} \left(x_{\consequence}|x_{\vari } \right) \rightarrow  \rho^{S } \left(x_{\consequence}|x_{\vari } \right)
$
 for all $x \in \mathcal{X}_{-0}$,
 then $ \frac{\rho(a, S_m)}{\rho(b,S_m)} \rightarrow \frac{\rho(a,S)}{\rho(b, S)}$. 
\end{ax*}

The Logit model is characterized by Luce's Choice Axiom \citep{luce1959individual}, which requires that $\frac{\rho(a, S')}{\rho(b,S')}=\frac{\rho(a,S)}{\rho(b, S)}$ whenever $a,b \in S \cap S'$.
LCI* requires that  Luce's Choice Axiom holds provided that the DM makes the same predictions for $a$ and $b$ when she faces either $S$ or $S'$.
To see why, note that a consultant represented by $R_Y$ thinks that the causal effect of own output on profit equals the distribution of profit conditional on own output.
If this distribution is the same for $\rho^{S'}$ and $\rho^{S}$, then she makes the same predictions about the effect of each workforce size when facing either $S'$ or $S$. Hence, she should rank the two the same, and so Luce's Choice Axiom should hold, regardless of whether her data comes from $\rho^{S'}$ or from $\rho^{S}$.
Indeed, letting $S_m=S'$ for all $m$, the axiom implies that  if $ \rho^{S'}(x_{\consequence}|x_{\vari }) =  \rho^{S}(x_{\consequence}|x_{\vari })$ for all $x \in \mathcal{X}_{-0}$, then  $\frac{\rho(a, S')}{\rho(b,S')} = \frac{\rho(a,S)}{\rho(b, S)}$.
Moreover, the Luce ratio should be ``almost'' the same whenever these conditional probabilities are ``close.''

Finally,  she behaves as if Logit-EU when her dataset is consistent with $R_{\vari }$.
\begin{ax*}[Correct Predictions Logit-EU$^*$, CPL*] 
There exists a increasing $u:\mathcal{X}_{\consequence}\rightarrow \mathbb{R}$  so that for every $a\in S$,
$$\rho(a,S)= \frac{\exp \left(  \sum_{x \in \mathcal{X}_{-0}} u(x_{\consequence})a(x) \right)}
{ \sum_{b\in S} \exp \left(  \sum_{x \in \mathcal{X}_{-0}} u(x_{\consequence})b(x)\right)}$$
whenever $b(y_{\consequence},y_{\vari })=b(y_{\vari })a(y_{\consequence}|y_{\vari })$ for all $y \in \mathcal{X}_{-0}$ and $a,b \in S$.
\end{ax*}

When profit is in fact independent of the size of workforce given own output, a DM represented by $R_Y$ makes correct predictions. 
Hence, she should behave according to Logit-EU. 
Notice that $X_{\vari} \perp_a X_{\consequence}$ and $X_{\vari} \perp_b X_{\consequence}$ do not imply that $X_{\vari} \perp_{\rho^{\{a,b\}}} X_{\consequence}$, so it is not sufficient that the conditional independence holds for each individual $b \in S$.%
\footnote{For one easy example, let $a(1_\vari,1_\consequence)=1$ and $b(0_\vari,0_\consequence)=1$. Clearly, $X_{\vari} \perp_a X_{\consequence}$ and $X_{\vari} \perp_b X_{\consequence}$ both hold, but $X_{\vari} \not \perp_{c} X_{\consequence}$ for  $c=\alpha a+(1-\alpha)b$ and any $\alpha \in (0,1)$.}

These axioms hold if and only if the consultant has an Endogenous SCR with DAG $R_{\vari }$.
\begin{cor}
The choice rule $\rho$ satisfies Full-support, Bounded Misperception, CRC*, I5*, LCI*, and CPL* if and only if $\rho$ has an Endogenous SCR $(R_{\vari },u,\consequence)$.
\end{cor}
The result is a corollary of Theorem \ref{thm: proper SCR}, and so we defer discussion until then.

\subsection{Foundations - general case}
\label{sec: general axioms}
We now present the axioms for an arbitrary DAG.

We adapt the identification results from Section \ref{sec:ID} to identify a candidate DAG. Recall that by Lemma \ref{lem: adjacent iff rmap}, if the DM is represented by the DAG $R$, then the variable $i$ is  $\rho$-adjacent to $j$ if and only if it follows or precedes it in some $R$-MAP. The definitions of minimal separators and $\rho$-adjacent (Definitions \ref{def: separator} and \ref{def: rho adjacent}) adjust naturally, replacing $q$ with $ {\rho^S}$ or $S$ as appropriate. These definitions require that two sets of variables are independent in the data generated by the DM's choices. We construct such menus as follows. For sets of variables $A,B$ such that $0\notin A$, $A\cap B=\emptyset$ and $a\in S$,  if $b(x)=a(x_A)b(x_B)b(x_{[A\cup B]^c}|x_{A\cup B})$ for all $b \in S$ and $x \in \mathcal{X}_{-0}$, then $X_A \perp_{\rho^S} X_B$ regardless of how the DM chooses.
The proofs of Lemmas \ref{lem: I contains AP iff block} and \ref{lem: adjacent iff rmap} each construct an $\{a,b\}$ for which $\rho^{\{a,b\}}$ satisfies the independence conditions in the definition. By Lemmas \ref{lem: I contains AP iff block}.(ii) and \ref{lem: adjacent iff rmap}.(ii), if  one $\{a,b\}$  reveals $i$ is $\rho$-adjacent to $j$ or that $A$ is a $\rho$-separator, then any $\{a,b\}$ satisfying the conditions does as well. 

The next axiom ensures that we can  consistently  infer the DM's DAG.

\begin{ax}[Consistent Revealed Causes, CRC]
\label{ax: Consistent Revealed Causes}
The set  $\mathcal{A}^\rho $ can be ordered $\mathcal{A}^\rho=\left\{A_1,\dots,A_{m} \right\}$, where $A_1 =\{0\}$ and $|A_m |= 1$, 
so that (i) $j$ is $\rho$-adjacent to $k$ if and only if there exists $i$ and $i'\in \{i-1,i+1\}$ for which $j \in A_i \setminus A_{i'}$  and $ k \in  A_{i'} \setminus A_{i} $  and (ii) if $j \in A_{i+1} \setminus A_i$, then $j \notin A_1 \cup \dots \cup A_i$.
\end{ax}

CRC ensures that we can make consistent inferences about the DM's DAG and the outcome. The candidate outcome is the (unique) variable in $A_m$, and the candidate  DAG $R$   has $ A_i \subset R(j) \subset A_i \cup A_{i+1} $ whenever $j \in A_{i+1} \setminus A_{i}$. For an intuition, suppose that the minimal separators are disjoint. Then by Lemma \ref{lem: adjacent iff rmap}, every MAP contains one variable from each $A \in \mathcal{A}^\rho$, 
and variables that follow each other in that MAP are adjacent. Thus, the minimal separator sets can be ordered using adjacency, which reveals that the DM's  DAG satisfies the above inclusions. 

In general, suppose that $\rho$ has an Endogenous SCR with a DAG $R$ having a minimal number of edges. Let $A_{m}=\{n^*\}$. Suppose that $A_{m'}$ is a minimal separator, and  that for every $i,j \in A_{m'}$, either $i R j$ or $j R i$. This is clearly true when $m'=m$.
For every $m'>1$, we recursively define $A_{m'-1}$ as follows.
Consider the node $k$ in $A_{m'}$ that does not cause any other node in $A_{m'}$; such a node exists because $R$ is acyclic. Let $B$ be the set of all nodes preceding $k$ on an $R$-MAP, and 
 $B'$ be the set of variables in $A_{m'}$ that are part of an $R$-MAP that does not intersect $B$.  Then, $A_{m'-1} = B\cup B'$ is a minimal separator, and since every $i' \in A_{m'-1}$ causes $k$, either $i R j$ or $j R i$ for every $i,j \in A_{m'-1}$ by perfection. Perfection also requires that $i R j$ for every
$i \in A_{m'-1}$ and  $j \in A_{m'} \setminus A_{m'-1} $; if in addition, $i \notin A_{m'}$ (so that $i \in B$) then $i$ precedes $j$ on an $R$-MAP.%
\footnote{If $i \in A_{m'-1}$, then  $i R k$ and $j R k$, so either $i R j$ or $j R i$. If also $j R i$, then one of $i$'s ancestors causes $j$ and so $j \in B'$.} Therefore, if $i \in A_{m'}\setminus A_{m'-1}$, then $A_{m'-1} \subset R(i) \subset  A_{m'}\cup A_{m'-1}$, and $i$ is $\rho$-adjacent to every $j \in A_{m'-1} \setminus A_{m'}$ by Lemma \ref{lem: adjacent iff rmap}.

We say that $\left(A_1,\dots,A_k\right)$ is an \emph{adjacent ordering} of $\mathcal{A}^\rho$ if it has the properties in the axiom's statement.

The fourth axiom requires that if the DM predicts that two actions lead to the same outcome distribution, then she chooses each with the same probability.

\begin{ax}[Indifferent If Identical Immediate Implications, I5]
\label{ax: same pr if first clique}For $a,b \in S \in \mathcal{S}$:\\
if $A \in \mathcal{A}^\rho$ is such that every $j\in A$ is  $\rho$-adjacent to $0$ and  $\marg_{A} a= \marg_{A} b$, then $\rho(a,S)=\rho(b,S)$.
\end{ax}

The covariates directly caused by the action are a sufficient statistic for the DM's prediction of the outcome distribution.
Whenever two actions lead to the same distribution over these covariates, she predicts that they lead to the same outcome distribution. She is therefore indifferent between any two actions with identical immediate implications.
The axiom only considers the marginal distributions on $X_{A}$ of $a$ and $b$, remaining agnostic about their distribution on other variables and any other available actions.

The fifth axiom ensures that if she makes similar predictions in two contexts, then she makes similar choices when facing them.

\begin{ax}[Luce's Choice Axiom Given Inferences, LCI]
\label{ax: luce axiom inference}
	For any $S,S_1,S_2, \dots \in \mathcal{S}$ with $a,b \in S_m \cap S$ for each $m$: 	if   $\left(A_1,\dots,A_k\right)$ is an adjacent ordering of $\mathcal{A}^\rho$ and
	$$\rho^{S_m} \left( x_{A_{i}  }|x_{A_{i-1} } \right) \rightarrow \rho^{S} \left(x_{A_{i}  }|x_{A_{i-1} } \right)$$    for every $x \in \mathcal{X}_{-0}$ and  $i =3,\dots,k$, 
	then $\frac{\rho(a, S_m)}{\rho(b,S_m)}\rightarrow \frac{\rho(a,S)}{\rho(b, S)}$.
\end{ax}
LCI relates the DM's inferences about the causal effects to her choices.
As per the discussion following CRC, if $\rho$ is represented by $R$ and  $j \in A_{i+1} \setminus A_i$, then $R(j) \subset  A_i \cup A_{i+1}$.
  Therefore, the hypothesis of the axiom implies that $\rho^{S_m} \left(x_j|x_{R(j)} \right) \rightarrow \rho^{S}\left(x_j|x_{R(j)} \right) $ for every $j \in A_3 \cup \dots \cup A_{k}$.%
  \footnote{Since $A_1=\{0\}$, $\rho^{S_m}(x_{A_2}|c)=\rho^{S}(x_{A_2}|c)$ for $c=a,b$ and all $m$, so this is trivially true for $j \in A_2$.}
Since her predictions about every causal effect is close, the Luce ratio of $a$ and $b$ from  $S_m$ should  be close to  their Luce ratio from $S$.
By taking $S_1=S_m$ for all $m$, we see that LCI implies that if her inferences are the same, then their Luce ratio is the same.

The next definition identifies a set of menus for which the DM's predictions about the outcome of each action are correct.

\begin{defn}\label{def: correct perception}
A menu $S \in \mathcal{S}$  \emph{leads to correct  predictions} if  
\[
b \left( x_{\cup_{i = 2}^{k} A_i} \right)=b \left(x_{A_2} \right)\prod_{i = 3}^{k} a \left(x_{A_{i} } |x_{A_{i-1}} \right) ,
\]
 for every $a,b \in S$ and $x\in \mathcal{X}_{-0}$ when $\left(A_1,\dots,A_k\right)$ is adjacent ordering of $\mathcal{A}^\rho$.
\end{defn}

By Theorem \ref{thm: RMAP characterize DAG}, only the variables indexed by $\bigcup_{i = 1}^{k} A_i$ are relevant for the DM's predictions. As per the discussion following CRC, the DM thinks that every $X_j$ with  $j \in A_{i+1}\setminus A_i$ is independent of $X_{\cup_{j = 1}^{i-1} A_j}$ conditional on $X_{A_i}$.
A correctly perceived menu satisfies these conditional independences, so her prediction of the distribution of all relevant variables, including the outcome, is correct.

The DM behaves in a standard fashion when her predictions are correct.
\begin{ax}[Correct  Predictions Logit-EU, CPL]
\label{ax: IA R markov}
\label{ax: markov domininace}
\label{ax: last}
There exists a strictly  increasing $u$ so that if $S$ leads to correct  predictions, then for any $a \in S$, $$\rho(a,S)=\frac{ \exp \left( \sum_{x\in \mathcal{X}_{-0}} u \left(x_{A_k} \right) a(x) \right)}{\sum_{b\in S}\exp \left( \sum_{x\in \mathcal{X}_{-0}} u \left(x_{A_k} \right) b(x) \right)}$$
where  $\left( A_1,\dots,A_k \right)$ is an adjacent ordering of $\mathcal{A}^\rho$.
\end{ax}
If $S$ leads to correct predictions, then every conditional independence assumption implied by her causal model holds in $\rho^S$. Hence,  her predictions about each of the actions are correct.
She should therefore behave according to Logit-EU. This axiom can be replaced by assuming that analogues of the independence, monotonicity, and continuity axioms hold for correctly perceived menus.

\begin{thm}
\label{thm: proper SCR}
A random choice rule $\rho$ has an Endogenous SCR if and only if $\rho$ satisfies Full-support, Bounded Misperception, CRC, I5, LCI, and CPL.
\end{thm}

The result highlights the connection between SCR and the Logit-EU model.
Notice that if $\mathcal{A}^\rho = \left\{ \{0\}, \{n^*\} \right\}$ and $\rho$ satisfies Axioms 1-\ref{ax: last}, then the choice rule has a Logit-EU representation.
The axioms relate deviations from Logit-EU to inconsistent predictions about causal effects.
I5 says that two alternatives are chosen with same probability whenever they coincide on the distribution of variables the action is revealed to cause, whereas Logit-EU requires coincidence on the outcome distribution.
LCI requires that the Luce ratio is constant only when predictions are constant, whereas Logit-EU requires it to be constant across all menus.  
CPL limits deviations from Logit-EU to situations where the DM's predictions do not match reality. 

We outline the proof for sufficiency here, and defer a formal proof to the appendix.
By CRC, we can define a candidate DAG $R$  and $n^*$  to represent $\rho$. 
Given $S \in \mathcal{S}$, we construct a correctly predicted $S_1$ so that for every $a \in S$, there is an $a' \in S_1$  with the same distribution that she would predict for $a$ given the candidate DAG, i.e.,  $a'(\cdot)=\rho^S_{R}(\cdot |a)$.
Our goal is to show that for any $a,b \in S$, the DM chooses $a$ and $b$  from $S$ with the same relative frequency  as  she chooses $a'$ and $b'$ from $S_1$.
To do so, we add distinct alternatives to $S_1$ to form a nested sequence $(S_m)_{m=1}^\infty$ where each $S_m$ is correctly predicted.
Bounded Misperception implies that the probability of choosing anything in $S$ from $S_m \cup S$ goes to zero as $|S_m|$ goes to infinity.
Therefore, the inferences that the DM makes when facing $S_m \cup S$ approach those she makes from $S_1$, which in turn equal those she makes from $S$.
LCI implies that the relative frequency with which $a'$ and $b'$ are chosen from $S_m \cup S$ converges to that of $a'$ and $b'$ from $S_1$.
Moreover, $a$ and $a'$ (as well as $b$ and $b'$) are chosen from $ S_m \cup S$ with the same probability by I5.
Applying LCI another time, we see that $a$ and $b$ are chosen with the same relative frequency in $S$ as $a'$ and $b'$ are in $S_1$, completing the proof.

\section{Discussion}
Theorem \ref{thm: RMAP characterize DAG} shows that one can use choice of actions  to reveal the relevant parts of a subjective causal model.
It suggests the types of questions, and the regularities in data, that could be used in surveys or experiments to infer or to test a subjective model.
The approach  provides both an alternative approach to asking subjects directly to describe their model, as in \cite{andre2021}, and a complementary way to test whether subjects actually use the model that they described. 

We conclude by discussing comparative behavior and alternative interpretations of the model.
\subsection{Comparative Coarseness}
A coarser causal model leaves out some variables or relationships relative to another.
Authors often explain ``irrational'' behavior in situations with adverse selection via coarseness.
For instance, \cite{EysterRabin2005}, \cite{Jehiel2008}, and \cite{Esponda2008} argue that the winner's curse reflects bidders who do not fully take into account the relationship between others' actions and signals.%
\footnote{Section 5 of \cite{Spiegler2016} discusses how and to what extent these models fit into the DAG framework.}
In this subsection, we compare DMs in terms of the coarseness of their model.
In particular, how can we separate two DMs who differ in that one's model contains more variables than the other's?

\begin{defn}
Say that $\rho_2$ \emph{has a coarser model than} $\rho_1$ if $X_i \perp_q X_{N \setminus \{i\}}$ for all $i \in N$ that are not  in any $R_2$-MAP from $0$ or an $R_2$-confounder to $n^*$, then for any $a \in S_q$, $\rho_1(a,q)\geq \rho_1(b,q) $  for all $b \in S_q \iff \rho_2(a,q)\geq \rho_2(b,q) $  for all $b\in S_q$.
\end{defn}

Consider DM1 represented by $\rho_1$ and DM2 represented by $\rho_2$.
As revealed by Theorem \ref{thm: RMAP characterize DAG}, the variables that do not belong to an $R_2$-MAP from $0$ or an $R_2$-confounder to $n^*$ are irrelevant for DM2's predictions.
The condition says that whenever those variables are independent of all other variables, i.e. they are actually irrelevant when forming predictions using any DAG, then the two DMs behave in the same way.
This ensures that if DM2 thinks a variable is relevant, so does DM1.

\begin{proposition}
\label{prop:coarseness}
Let $\rho_i$  have an SCR $(R_i,u_i,n^*)$ for $i=1,2$.
Then,  $\rho_2$ has a coarser model than $\rho_1$ if and only if  there exists $N' \subset N$ so that $\rho_2$ has an SCR $\left ( R_1 \cap N' \times N',u_1,n^* \right)$.
\end{proposition}

The comparison reveals when the models of two DMs are nested.
Specifically, they agree on the causal relationship between any two variables that both consider relevant and on the desirability of outcomes.
However, more variables may be relevant for DM1's predictions  than for DM2's. 

\subsection{Interpretations of the model}\label{sec: interpretations}
Our main interpretation of an SCR is the one discussed above: 
it describes a DM who predicts the outcome of her action using a causal model and adjusts her predictions for the spurious correlation caused by any confounders.

Alternatively, the model may also describe a DM with limited data access \citep{Spiegler2017}.
In this interpretation, she only considers or observes the distributions of several overlapping subsets of variables.
She then extrapolates from the partial datasets  using the distribution that maximizes entropy subject to matching the marginal distributions over the datasets. 
Identifying her DAG corresponds to identifying the considered subsets.
In the example, a consultant represented by $R_{\vari }$ only observes, or only has access to, two datasets;
one that keeps track of the effect of the workforce size on own output and another one that contains the correlations between own output, competitor output and profit. 

Another interpretation is a particular kind of correlation neglect caused by limited memory.
Without confounders, the DAG reduces the number of parameters needed to reconstruct the joint distribution over the variables.
In the running example, a DM with DAG $R_{\covariate }$ can store all the information that she deems relevant  using only $6$ parameters when all variables are binary. It would require $2^4-1=15$ parameters to record the probability of each possible realization.

Another interpretation is that the DM estimates a structural equation model (SEM). 
She estimates each variable $X_i$ based on the variables indexed by $R(i)$. 
Her estimates have a causal interpretation if she includes all the variables (and only those) that she believes have a direct causal impact on it. See Chapter 5 of \cite{Pearl2009} for a more in-depth discussion of this. 
 
For a final interpretation, we note that when $\rho$ has an  SCR $(R,u,n^*)$, $q_R$ minimizes Kullback-Liebler divergence from $q$ among all the probability distributions on $\mathcal{X}$ that are consistent with $R$.%
 \footnote{See Section 5.5 of \cite{Hajeketal1992}.}
A $\rho$ with an Endogenous SCR is a single agent Berk-Nash equilibrium \citep{EspondaPouzo2016Berk} with extreme-value errors.
As in that model, we can interpret the behavior as the steady state of a learning process with a set of parameters (probability distributions) that does not include the  ``true'' one.
\newpage

\appendix
\section{Proofs omitted from the main text}
\subsection{Notation}
For a DAG $Q$, let $\tilde{Q}$ be the skeleton or undirected version of $Q$, i.e. $j \tilde{Q} i$ if and only if either $i Q j$ or $jQi $, $(i,j,k)$ be a $Q$-v-collider if and only if $i Q k$, $j Qk$, $i \!\not\!\! Q j$ and $j \!\not\!\! Q i$, and
$Q^*$ be the DAG that drops all edges into $0$ so $p_Q(x|do(a))=p_{Q^*}(x|a)$ for all full support $p$.
 
\begin{defn}
Let $\mathcal{C}$ be a collection of subsets of a finite set and $\mathcal{T}$ a tree with $\mathcal{C}$ as its node set. 
Say that $\mathcal{T}$ is a \emph{junction tree} if for any $C_1,C_2 \in \mathcal{C}$, $C_1\cap C_2$ is contained in every node on the unique path in $\mathcal{T}$ between $C_1$ and $C_2$.  
\end{defn}
The set $C \subseteq \{0,\dots,n\}$ is a clique for $R$ if $j \tilde{R} k$ for all $j,k \in C$.
By Theorem 4.6 of \cite{cowell1999}, the maximal cliques of a   perfect  DAG $R$ can be linked to form a junction tree.
Call this a \emph{maximal clique junction tree (MCJT)} for $R$.

\subsection{Preliminary Lemmas  for Theorems \ref{thm: RMAP characterize DAG} and \ref{thm: proper SCR}}
Throughout, there is no loss in considering only an $R$ for which $R \cap An_{R^*}(n^*)^2=R$. Recall that as a binary relation, $R \subset N\times N$ and $R \cap E \times E = R \cap E^2$ with $E\subset N$ defines a new binary relation that retains only the edges between nodes in $E$.

Let $C^*_R$ (or just $C^*$ when $R$ is clear) be the set consisting of $0$ and all $R$-confounders. Notice $C^*_R$ is a clique: for $j,k \in C^*_R$, $j \tilde{R} k$ if either $j=0$ or $k=0$ and if $j,k\neq 0$, then  $j R 0$ and $k R 0$ so $j \tilde{R} k$ by perfection. 

First, we show that we can reduce the number of nodes in $R$ without changing predictions. Set $R^0=R$, $N^0(R)= An_R(n^*)$. 
Inductively define $N^{i+1}(R)$ and $R^{i+1}$ as follows.
There is a unique maximal $R^i$-clique $C^{i\dag}$ containing $n^*$, since $j R^i n^*$ and $k R^i n^*$ implies $j \tilde{R}^i k$ by perfection and $n^* \notr^i j$ for all $j$.
For any MCJT for $R^i$, there is a unique path from any maximal clique containing $C^*$ to $C^{i\dag}$. 
Take a maximal clique containing $C^*$ and the MCJT with one of the shortest such paths, and let $\mathcal{C}^{i}=\{C^i_1,\dots,C^i_m\}$ be the maximal cliques in this path ordered so that so that there is an edge from $C^i_j$ to $C^i_{j+1}$, and $C^*\subseteq C^i_1$.
If $\cup \mathcal{C}^i \neq N^{i}(R)$, then set $N^{i+1}(R)= \cup \mathcal{C}^i$.
If $\cup \mathcal{C}^i = N^{i}(R)$ and there is $j_{i+1} \in N^{i}(R) \setminus [C^* \cup \{n^*\} ]$ that is not contained in at least two maximal $R^{i}$-cliques, then set $N^{i+1}(R)= N^{i}(R) \setminus j_{i+1}$.
Otherwise, $N^{i+1}(R)=N^i(R)$.
Set $R^{i+1}=R^i \cap [N^{i+1}(R) \times N^{i+1}(R)]$.

Since there are finite variables and cliques, there is $\bar{i}$ so that $N^{i+1}(R)=N^i(R)$ for all $i>\bar{i}$. Let $N^*(R) =N^{\bar{i}+1}(R)$. 
We show that edges between variables not in $N^*(R)$ can be dropped.

\begin{lemma}
\label{claim_fundamental}
If $R$ is a perfect non-trivial DAG and $Q =R\cap N^*(R)^2$, then for all $x \in \mathcal{X}$, $p_{R^*}(x_n^*|x_{C^*})=p_{Q^*}(x_n^*|x_{C^*})$.
\end{lemma}

\begin{proof}
Clearly, $p_{R^{0*}}(x_{n^*}|x_{C^*})=p_{R^{*}}(x_{n^*}|x_{C^*})$.
Assume (IH) that $p_{R^*}(x_{n^*}|x_{C^*})=p_{R^{i*}}(x_{n^*}|x_{C^*})$. We show that (IH) implies $p_{R^*}(x_{n^*}|x_{C^*})=p_{R^{i+1*}}(x_{n^*}|x_{C^*})$. This establishes the result.

If $\cup \mathcal{C}^i \neq N^{i}(R)$, then  by properties of the MCJT, we can treat $C^i_1$ as ancestral and all the cliques outside of $\mathcal{C}^{i}$ as being non-ancestors of $n^*$. Hence $$p_{R^{*}}(x_{n^*}|x_{C^*})=p_{R^{i*}}(x_{n^*}|x_{C^*})=p_{R^{i+1*}}(x_{n^*}|x_{C^*}).$$ That is, we can drop all variables that do not appear in at least one clique in $\mathcal{C}^i$. 

If $\cup \mathcal{C}^i = N^{i}(R)$, then $(C^i_1,\cdots, C^i_m)$ is a MCJT for $R^{i}$. By construction, it has a single path that can be ordered naturally and, without loss, has $C^* \subset C^i_1$ and $n^* \in C^{i}_m $.
Let $k$ be so that $j_{i+1} \in C^i_k$ and $j_{i+1} \notin C^{i}_{k'}$ for all $k' \neq k$.

 Let $l(k')<k'$ be the highest index for each $k'>1$ so that $C^{i}_{k'} \cap C^{i}_{k'-1} \subset C^{i}_{l(k')}$.
Define $Q^i$ so that $jQ^i j' \Leftrightarrow jR^i j'$ when $j,j'\in C^*$, $jQ^i j'$ whenever $j\in C^*$ and $j' \in C_1^i \setminus C^*$, $jQ^i j'$ whenever $j \in C^{i}_{k'}\cap C^{i}_{l(k')}$ and $j' \in C^{i}_{k'}\setminus C^{i}_{l(k')}$, $jQ^i j_{i+1}$ whenever $j\in C^{i}_{k}\setminus\{j_{i+1}\}$, and $j \tilde{Q}^i j'$ for each distinct $j,j' \in C^{i}_{k'}$ for all $k'$. 

Notice $Q^{i*}$ has the same skeleton and the same v-colliders as $R^{i*}$ and that  $j_{i+1} \notin An_{Q^i}(n^*)$. Hence $$p_{[Q^{i} \cap N^{i+1}(R)^2]^*}(x_{n^*}|x_{C^*})=p_{Q^{i^*}}(x_{n^*}|x_{C^*})=p_{R^{i*}}(x_{n^*}|x_{C^*})=p_{R^*}(x_{n^*}|x_{C^*})$$
where the last equality holds by hypothesis and the second by   Theorem 1 of \cite{Verma1991equivalence}.
Moreover, $[Q^{i} \cap N^i(R)^2]^*$ has the same skeleton and same v-colliders as $R^{i+1*}$, so $p_{R^*}(x_{n^*}|x_{C^*})=p_{R^{i+1*}}(x_{n^*}|x_{C^*})$.
\end{proof}

By Lemma \ref{claim_fundamental}, a MCJT for a perfect $R$ such that $R \subset N^*(R)^2$ consists of single path. We say that $(C_1,\dots,C_m)$ is the \emph{Canonical Maximal Clique Junction Tree (CMCJT)} for such an $R$ if $C_i$ is a maximal $R$-clique for each $i$, $\{C_1,\dots,C_m\}$ is a MCJT for $R$ that has an edge between $C_i$ and $C_{i+1}$ for each $i$, and $C^*_R \subseteq C_1$. 

Thus, every non-redundant variable appears in at least two  maximal cliques. Therefore, each maximal clique is the union of its intersections with its neighbors. Below, we show that the intersection of adjacent maximal cliques corresponds to a $\subseteq$-minimal $\rho$-separator.

\begin{lemma}
\label{lem: clique union of intersection}
If $(C_1,\dots,C_m)$ is the CMCJT for a perfect and non-trivial DAG $R \subset N^*(R)^2$, $A_0=C^*$, $A_m=\{n^*\}$, and $A_i =C_i \cap C_{i+1}$ for $i=1,\dots,m-1$, then $C_i = A_{i-1} \cup A_i$.
\end{lemma}

\begin{proof}
First note that this clearly holds for $i=1,m$. Thus, pick an $i\in \{2,\cdots,m-1\}$ and $j \in C_i$.   By Lemma \ref{claim_fundamental}, there exists $i' \neq i$ so that $j \in C_{i'}$. If $i'>i$, $j \in C_{i+1}$ and hence $j \in A_{i+1}$. If $i'<i$, $j \in C_{i-1}$ and hence $j \in A_{i-1}$.
\end{proof}

Finally, we relate the links in $R$ that cannot be reversed to the maximal cliques and thus the $\subseteq$-minimal $\rho$-separators.
\begin{defn}
\label{def: char fundamental}
For a DAG $Q\subseteq N^*(Q)^2$, $k \hat{Q} j$  if and only if $k Q j$ and either:
\vspace{-1em}
\begin{itemize}
\item [(a)] $k\in C^*$ and either $j \notin C^*$ or $j=0$, 
\item [(b)] $k$ precedes $j$ on a $Q$-MAP from $l\in C^*$ to $n^*$, or
\item [(c)] there exists a node $l\in N$ such that $l \hat{Q} k$ and $l \not \!\! Q j$.
\end{itemize}
\end{defn}

\begin{lemma}
\label{lem: fun links and cliques}
If $(C_1,\dots,C_m)$ is the CMCJT for a perfect and non-trivial DAG $R \subset N^*(R)^2$ and $C_0=C^*_R \subset C_1$ and $\{k,j\} \not\subseteq C_0$, then $k \hat{R} j $ if and only if there exists $i$ so that $k \in C_i \cap C_{i-1}$ and $j \in C_i \setminus C_{i-1}$.
\end{lemma}

\begin{proof}
First we show necessity.
Let (*) be the assertion that ``if $k\in C_{i} \cap C_{i-1}$ and $j\in C_{i}\setminus C_{i-1}$, then $k \hat{R}j$.'' We show that (*) holds for all $i$ by induction. For $i=1$, $k \in C_0 \cap C_{1}$ iff $k \in C^*$. If $j \in C_1 \setminus C_0$, then $k R j$. To see this, suppose not, i.e. $k\notr j$. As $j,k\in C_1$ this implies that $jRk$. If $j R 0$, then $j \notin N^*(R)$ since $j \notin C^*$ and $j R 0$ imply all paths from $j$ to $n^*$ go through $C^*$. Hence $0 R j$ and $k \neq 0$, but $j R k$ implies  $0 R j R k R 0$,  a cycle. Conclude $k R j$.

Assume  (IH) that (*) holds for all $i'\in \{1,\cdots,i-1\}$ where $i \geq 2$.
Take any $k \in C_i \cap C_{i-1} $ and $j \in C_i \setminus C_{i-1}$.
We show first that $k R j$. If not, then $j R k$ since $j,k \in C_i$. Pick any $l \in C_{i-1} \setminus C_{i-2}$ and $l' \in C_{i-1} \setminus C_i$. Since $l'\in N^*(R)$, we have $l'\in C_{i-1}\cap C_{i-2}$, and hence by (IH) $l'\hat{R} l$. We must have $l R j$. To see this, note that by Lemma \ref{lem: clique union of intersection}, $l \in C_i$ and $l' \in C_{i-2}$. By definition of a junction tree, $l'\!\not\!\! \tilde{R} j$.  $j R l$ would imply $j \tilde{R} l'$. 
Because $l',k \in C_{i-1}$, $l' \tilde{R} k$. If $k R l'$, then $k R l' R l R j R k$ is a cycle in $R$. If $l' R k$, then $j \tilde{R} l'$ by perfection of $R$. Either is a contradiction, so $k R j$ must hold.

We now show that $k \hat{R} j$. If $k \in C_0$, then $k \hat{R} j$ follows from Definition \ref{def: char fundamental}.(a). Otherwise, there is $i'< i-1$ so that $k \in C_{i'} \setminus C_{i'-1}$. For any $l \in C_{i'} \setminus C_{i'+1}$, $l \in C_{i'-1}$ by Lemma \ref{lem: clique union of intersection} and so $l \hat{R} k$ by IH. But $l \not \tilde{R} j$, so $k \hat{R} j$ by \ref{def: char fundamental}.(c).


 To complete the proof, we show that ``if $j\in C_i\setminus C_{i-1}$ and $k\hat{R}j$, then $k\in C_{i-1}\cap C_i$.'' 
Let $j \in C_i \setminus C_{i-1}$ and $k \hat{R} j$.
If $i=1$ and $k \in  C^*$, then $k \in C_0 \cap C_1$.
If $i=1$ and $k \notin  C^*$, then $j,k \in C_1 \setminus C_0$. Therefore, $k' R j$ and $k' R k$ for all $k' \in C^*$ by necessity. But this rules out Definitions 1.(a)-(c) holding for $k$ and $j$, since no $R$-MAP can contain both. 
For $i > 1$,   take any $j' \in C_{i}  \setminus C_{i+1}$, noting  $j' \in C_{i-1} \cap C_i$ by Lemma \ref{lem: clique union of intersection}. By necessity, $j' \hat{R} j$. If $j'=k$, then we are done. 
If not, then $k \tilde{R} j'$ by perfection, so $k \in C_{i'}$ for some $i'\leq i$. 
Similarly, $k \tilde{R} j$ so $k \in C_{i''}$ for some $i'' \geq i$. 
Because $(C_1,\dots,C_m)$ is the CMCJT, $k  \in C_i $. 
To see $k \in C_{i-1} \cap C_i$, suppose not, so $k \notin C_{i-1}$.
Hence $k\not \in C^*$. 
 Since any path from a node in $C^*$ to $n+1$ contains some $j' \in C_{i-1} \cap C_i$, and $j' \hat{R} j$ and $j' \hat{R} k$  by necessity, no $R$-MAP includes both $j$ and $k$.
Therefore, by Definition \ref{def: char fundamental} there is $l$ so that $l\hat{R}k$ and $l\notr j$.
Then, $l \in C_i$, since otherwise $k \hat{R} l$ by necessity. Thus, $jRl$.
But then $jRlRkRj$, a contradiction. Hence $k\in C_{i-1}\cap C_i$.
\end{proof}

Finally, we show that all other links in $R$ prevent v-colliders.
\begin{lemma}\label{lem:R not E implies v}
For a perfect and non-trivial DAG $R \subseteq N^*(R)^2$, if $j R k$ and $j\ \hat{\notr} k$, then there exists $l$ such that $j \hat{R} l$ and $k \hat{R} l$.
\end{lemma}

\begin{proof}
Pick any $j,k\in N^*(R)$ so that $j R k$ and $j \ \hat{\notr} k$.
If $\{j,k\} \subset C^*$, then $k\neq 0$ by Definition \ref{def: char fundamental}(a), and $j\neq 0$ as $k\in C^*$. Thus, $j$ and $k$ are $R$-confounders, and so $j \hat{R} 0$ and $k \hat{R} 0$ by Definition \ref{def: char fundamental}.
Otherwise, let $(C_1,\dots,C_m)$ be the CMCJT for $R $, noting that  $j,k \in C_i  $ for some $i$ and that if $k \notin C_{i-1}$, then $j \notin C_{i-1}$ by Lemma \ref{lem: fun links and cliques} and $j     \ \hat{\notr}  k$. This combined with Lemma \ref{lem: clique union of intersection} implies that either $j,k \in C_i \cap C_{i-1}$ or $j,k \in C_{i} \cap C_{i+1}$. In the former case pick $l \in C_{i} \setminus C_{i-1} $ and in the latter pick $l \in C_{i+1} \setminus C_{i} $. In either case, $j \hat{R}l$ and $k \hat{R} l$  by Lemma \ref{lem: fun links and cliques}.
\end{proof}

The following rewriting of $p_{R^*}(x)$ will be useful.
\begin{lemma}
\label{lem: JT rearrange}
If $(C_1,\dots,C_m)$ is the CMCJT  of a   perfect and non-trivial   DAG $R  \subset  N^*(R)^2$, then for any $p \in \mathcal{D}$,
\begin{equation}
p_{R^*}(x_{N^*(R)})=p(x_{C^* \setminus \{0\}})p(x_0)p(x_{C_1}|x_{C^*})\prod_{j=1}^m p( x_{C_j} | x_{C_j \cap C_{j-1}}).
\label{eq: JT}	
\end{equation}
\end{lemma}

\begin{proof}
Let $p^*$ be defined by the RHS of Equation (\ref{eq: JT}) where $(C_1,\dots,C_m)$ is the CMCJT of $R$.
Let $\bar{C}_{k}= \cup_{i=1}^k C_i$.
Proceed by induction to show that $p^*(x_{\bar{C}_k})=p_{R^*}(x_{\bar{C}_k})$.
Since $C^* \setminus \{0\}$ and $\{0\}$ are ancestral cliques in $R^*$ that are $R^*$-independent, $p_{R^*}(x_{C^*})=p(x_{C^* \setminus \{0\}})p(x_0)$.
That $p^*(x_{\bar{C}_1})=p_{R^*}(x_{\bar{C}_1})$ follows from the formula for $p_{R^*}(x_{\bar{C}_1})$. 

Suppose (IH) that $p^*(x_{\bar{C}_{M-1}})=p_{R^*}(x_{\bar{C}_{M-1}})$ and that if $i \in C_{M-1}\setminus C_{M-2}$, then $i' R i$ for all $i' \in C_{M-1} \cap C_{M-2}$. 
Notice the $C_M \cap \bar{C}_{M-1}=C_M \cap C_{M-1}$ by definition of CMCJT. 
Take any $l \in C_{M-1} \setminus C_M$. By Lemma \ref{lem: clique union of intersection}, $l \in C_{M-1} \cap C_{M-2}$, and by definition of CMCJT, $l \notin C_{m'}$ for any $m'>M-1$.
First, by Lemma \ref{lem: fun links and cliques}, $k R j$ for all $k \in C_M \cap C_{M-1}$ and $j \in C_M \setminus C_{M-1}$.
Second, we show that  $k \notr j$ for all $ k \notin \bar{C}_M$. If $k R j$ for some $k \notin \bar{C}_M$, then $k \tilde{R} l$ by perfection, contradicting that $l \notin C_{M}$ or that $k \notin \bar{C}_M$.

Pick $j_1 \in C_M \setminus C_{M-1}$ so that $j_1 R j'$ for all $j' \in C_M \setminus  C_{M-1}$. $j_1$ exists since $R$ is acyclic and $C_M \setminus C_{M-1}$ is a clique. By the above, $R(j_1)= C_M\cap C_{M-1}$. Now, $$p_{R^*}(x_{\bar{C}_{M-1}},x_{j_1})=p^*(x_{\bar{C}_{M-1}})p(x_{j_1}|x_{R(j_1)})=p^*(x_{\bar{C}_{M-1}})p^*(x_{j_1}|x_{C_M\cap C_{M-1}}).$$
For $k \geq 2$, recursively define $j_k\in  C_M \setminus [ C_{M-1} \cup \{j_1,\dots,j_{k-1}\}]$ so that $j_k R j'$ for all $j' \in  C_M \setminus [ C_{M-1} \cup \{j_1,\dots,j_{k-1}\}]$. $j_k$ exists by acyclicity.  Assume that \[p_{R^*}(x_{\bar{C}_{M-1}},x_{j_1},\dots,j_{k-1})=p^*(x_{\bar{C}_M-1})p (x_{j_1},\dots,x_{j_{k-1}}|x_{C_M\cap C_{M-1}}) \] and that $R(j_{k'})= \{j_1,\dots,j_{k'-1}\} \cup [C_{M-1} \cap C_M]$ for $k'<k$, which is true for $k=2$. By the above, $R(j_k) = \{j_1,\dots,j_{k-1}\}\cup [C_M\cap C_{M-1}]$. Therefore,
\begin{align*}
&p_{R^*}(x_{\bar{C}_{M-1}},x_{j_1},\dots,x_{j_{k}})\\
=& p^*(x_{\bar{C}_{M-1}})p (x_{j_1},\dots,x_{j_{k-1}}|x_{C_M\cap C_{M-1}})p(x_{j_k}|x_{R(j_k)}) \\
=& p^*(x_{\bar{C}_{M-1}})p(x_{j_1},\dots,x_{j_{k-1}}|x_{C_M\cap C_{M-1}})p(x_{j_k}|x_{j_1},\dots,x_{j_{k-1}},x_{C_M\cap C_{M-1}})\\
=& p^*(x_{\bar{C}_{M-1}})p (x_{j_1},\dots,x_{j_{k}}|x_{C_M\cap C_{M-1}})
\end{align*}
Applying enough times, we have $p^*(x_{\bar{C}_M}) = p_{R^*}(x_{\bar{C}_{M-1}}) p(x_{C_M \setminus C_{M-1}}|x_{C_M \cap C_{M-1}}) = p_{R^*}(x_{\bar{C}_M})$. 
Inductively,  $p^*=p_{R^*}$.
\end{proof}

\subsection{Sufficiency for Theorem \ref{thm: RMAP characterize DAG}}
Suppose that $\rho$ has a perfect SCR $(R, u, n^*)$ and $R'$ is   perfect and non-trivial DAG so that the set of $R$-confounders equals the set of $R'$-confounders, the set of $R^{\prime}$-MAPs from an $R$-confounder or $0$ to $n^*$ equals the set of $R$-MAPs from an $R$-confounder or $0$ to $n^*$.
By Lemmas \ref{lem: fun links and cliques} and \ref{lem:R not E implies v}, any two DAGs that have the same confounders and the same MAPs from confounders or $0$ to $n^*$ correspond to the same CMCJT after removing unnecessary variables. Therefore,  $q_{R^*}(x_{n^*}|a)=q_{R^{\prime*}}(x_{n^*}|a)$  for any $q$ and $a \in S_q$.

\subsection{Necessity for Theorem \ref{thm: RMAP characterize DAG}}
\label{sec: Proof Necessity ID}
It will be useful to introduce the following class of datasets.
Partition every $\mathcal{X}_{j'}$ into $\{\bar{E}_{j'},\underline{E}_{j'}\}$ where $x_{j'} \in \bar{E}_{j'}$ and $x'_{j'} \in \underline{E}_{j'}$ implies $x_{j'}>x'_{j'}$ and let $\bar{E}_0=\{a\}$ and $\underline{E}_0=\{b\}$.
A dataset $q$ is \emph{pseudo-binary} if $S_q=\{a,b\}$ and the distribution within each $\bar{E}_{j'}$ and $\underline{E}_{j'}$ is independent of all  events  outside of $j'$.
A pseudo-binary dataset has $\{i,j\}$-MLRP if and only if $q(\bar{E}_i|\bar{E}_j)>q(\bar{E}_i|\underline{E}_j)$ because 
$$\frac{q(\bar{E}_i|\bar{E}_j)}{q(\underline{E}_i|\bar{E}_j)}
>\frac{q(\bar{E}_i|\underline{E}_j)}{q(\underline{E}_i|\underline{E}_j)}\iff \frac{q(\bar{E}_i|\bar{E}_j)}{1-q(\bar{E}_i|\bar{E}_j)}
>\frac{q(\bar{E}_i|\underline{E}_j)}{1-q(\bar{E}_i|\underline{E}_j)}  $$
and the likelihood ratios are the same for any values in the same cell.
Also note $\{i,j\}$-MLRP is symmetric in $i$ and $j$.

\begin{proof}[Proof of Lemma \ref{lem: I contains AP iff block}]
[(i) $\implies$ (ii)] Suppose that $A$ intersects all paths from $0$ to $n^*$ and let $R'$ remove all links from $A$ to $ A^c$ and vice versa.
Observe that when $X_A \perp_q X_{A^c}$, 
$$q(x_{i}|x_{R(i)})=q(x_{i}|x_{R(i)\cap A},x_{R(i) \setminus A})=q(x_{i}|x_{R(i)\cap A})=q(x_i|x_{R'(i)})$$
when $i \in A$, and for $i \notin A$,
$$q(x_{i}|x_{R(i)})=q(x_{i}|x_{R(i)\cap A},x_{R(i) \setminus A})=q(x_{i}|x_{R(i)\setminus A})=q(x_i|x_{R'(i)}).$$
Hence $q_R(x)=q_{R'}(x)$ for all $x$. Since $0 \notin An_{R^{\prime}}(n^*)$, $X_0 \perp_{q_{R'}} X_{n^*}$  and so $\rho(a,q)=\rho(b,q)$ for all $a,b \in S_q$.

[(ii) $\implies$ (i)] By contrapositive. Let $(i_0,\dots,i_m)$ be an $R$-MAP from $0$ to $n^*$ that does not intersect $A$.		
Pick a pseudo-binary dataset $q$ that factorizes as 
\[q(x)=q(x_0)\prod_{j =1}^m q(x_{i_{j}}|x_{i_{j-1}}) \prod_{k' \notin \{i_0,\dots,i_m\}} q(x_{k'})\]
and where $q(\bar{E}_{i_{j}}|\bar{E}_{i_{j-1}})>q(\bar{E}_{i_{j}}|\underline{E}_{i_{j-1}})$ for all $j$.
This implies that $X_{A} \perp_q X_{A^c}$ since $A \subseteq N \setminus  \{i_0,\dots,i_{m}\}$. For such a dataset $q$, $\rho(a,q)>\rho(b,q)$ since   $q_{R }(\cdot|do(a))$ FOSD $q_{R }(\cdot|do(b))$.

[(ii) $\implies$ (iii)] This is immediate.

[(iii) $\implies$ (i) when $R$ unconfounded] Suppose not, so there is $A$ that does not intersect at least one $R$-MAP and a $q^A$ having the claimed properties for which $\rho(a,q^A)=\rho(b,q^A)$. Let $(C_1,\dots, C_m)$ be a CMCJT representation for $R\cap N^*(R)^2$ with $0\in C_1$ and $n^*\in C_m$, and $B_k=C_k\setminus A$ for all $k$.
Since $A$ does not intersect at least one $R$-MAP, $B_k\cap B_{k-1}\neq \emptyset$ for every $k$.

By Lemma \ref{lem: JT rearrange} and $X_A\perp_{q^A} X_{A^c}$, we note that
$$q^A_{R}(\bar{x}_{n^*}\mid a) \geq  q^{A}(\bar{x}_{B_1\setminus\{0\}}\mid a) \prod_{k=2}^{m} q^{A}(\bar{x}_{B_k\setminus B_{k-1}}\mid \bar{x}_{B_k\cap B_{k-1}}).$$
When $\delta=\left((2^{1/n} -1){q^A}(a)+1 \right)^{-1}$,
\begin{align*}
q^{A}(\bar{x}_{B_k\setminus B_{k-1}}| \bar{x}_{B_k\cap B_{k-1}})&=\frac{q^{A}(\bar{x}_{B_{k}}| a) q^{A}(a) + q^{A}(\bar{x}_{B_k}| b) q^{A}(b)}{q^{A}(\bar{x}_{B_k\cap B_{k-1}}|a)q^{A}(a)+q^{A}(\bar{x}_{B_k\cap B_{k-1}}| b)q^{A}(b)}\\
&>\frac{(\delta) q^{A}(a) + (0) q^{A}(b)}{q^{A}(\bar{x}_{B_k\cap B_{k-1}}|a)q^{A}(a)+q^{A}(\bar{x}_{B_k\cap B_{k-1}}| b)q^{A}(b)}\\
&>\frac{(\delta) q^{A}(a) }{(1) q^{A}(a)+(1-\delta)q^{A}(b)}= \frac{\delta  q^{A}(a)}{1-(1-q^{A}(a))\delta }= 2^{-1/ n},
\end{align*}
since $\delta \leq q^{A}(\bar{x}_{B_k}|a) < 1$, $0 < q^{A}(\bar{x}_{B_k}|b)\leq 1 - \delta$, and $q^{A}(b)=1-q^{A}(a)$. As there are no more than $n$ minimal separator sets, $q_{R}(\bar{x}_{n^*}|a)> \left(2^{-1/n}\right)^n =\frac12$.
Symmetrically, $q_{R}(\underline{x}_{n^*}|b)>\frac12$. Therefore, $U_{q^A}(a) > \frac12 u(\bar{x}) + \frac12 u(\underline{x})  > U_{q^A}(b)$ so $\rho(a,q)>\rho(b,q)$.
\end{proof}

We now relate $\mathcal{A}^\rho$ to the CMCJT of a DAG $R$ that represents $\rho$.
\begin{lemma}
\label{lem: separators intersection}
Let $\rho$ have a perfect SCR $(R\subseteq N^*(R)^2, u, n^*)$ and $R'=R\cap  \left[ N^* \left(R\setminus [R(0)\times N] \right) \right]^2$. Then, $A \in \mathcal{A}^\rho$ if and only if $A = C_i \cap C_{i+1}$ for some $i$ where $(C_1,\dots,C_{m})$ is the CMCJT for $R'$, $C_0 = \{0\}$, and $C_{m+1}=\{n^*\}$.
\end{lemma}

\begin{proof}
Let $\rho$ have a perfect SCR $(R\subseteq N^*(R)^2, u, n^*)$, $$R'=R\cap  \left[ N^* \left(R\setminus [R(0)\times N] \right) \right]^2,$$ $(C_1,\dots,C_{m})$ be the CMCJT for $R'$, $C_0 = \{0\}$,  $C_{m+1}=\{n^*\}$, $A_i=C_i \cap C_{i+1}$ for each $i$, and $B_i=\cup_{j=1}^i C_j$ for each $i$.
Moreover, observe that every $R$-MAP from $0$ to $n^*$ is an $R'$-MAP, and every $R'$-AP ($R'$-MAP) from $0$ to $n^*$ is a $R$-AP ($R$-MAP).
Clearly, $A_0,A_{m} \in \mathcal{A}$.

Pick any $i\geq 0$. We first show that every $R$-AP from $0$ to $n^*$ intersects $A_{i+1}$. Take any such $R$-AP. It contains an $R$-MAP $(i_0,\dots,i_M)$. Let $i_k$ be first index so that $i_k \notin B_{i}$. In particular, $i_k \in C_j \setminus C_{j-1}  $ for some $j \geq i+1$. Then, $i_{k-1} \in C_{j-1} \cap C_j \subset B_{j-1}$ by Lemma \ref{lem: fun links and cliques}, so $j-1=i$ and $i_k \in A_{i+1}$. By Lemma \ref{lem: I contains AP iff block}, $A_{i+1}$ separates.
	
Let $A \in \mathcal{A}\setminus \{\{0\},\{n^*\}\}$. By Lemma \ref{lem: I contains AP iff block}, $A$ intersects every $R$-AP from $0$ to $n^*$. By Theorem 4.4 of \cite{cowell1999}, $A$ is a clique, so $A \subset C_i$ for some $i$. 
We show that either $A=A_i$ or $A= A_{i-1}$. Since both $A_i$ and $A_{i-1}$ separate and $A$ is minimal, it suffices to show that either $A_{i-1}\subseteq A$ or $A_i\subseteq A$. 

For contradiction, suppose that  $A_i \cap A_{i-1} \not \subseteq A$. There exists $j \in A_i \cap A_{i-1} \setminus A$. Let $i'<i-1$ be such that $j \in C_{i'} \setminus C_{i'-1}$. Then, there exist $j' \in C_{i'} \setminus C_{i'+1} \subseteq C_{i'} \cap C_{i'-1}$ and $l \in C_{i+1}\setminus C_i$. By Lemma \ref{lem: fun links and cliques},  $j' \hat{R}' j$ and $j \hat{R}' l$, so there exists a $R$-AP that does not intersect $A$.
Therefore, $A_i \cap A_{i-1}  \subseteq A$.

If $A_{i-1}  \subseteq A_i$ or $A_i  \subseteq A_{i-1}$, then we are done since $ A_i \cap A_{i-1}  \subseteq  A$. Otherwise, for any  $j \in A_{i-1} \setminus A_i$ and any $k \in A_{i} \setminus A_{i-1}$, $j \hat{R} k$ by Lemma \ref{lem: fun links and cliques}, and by the above, there is an $R$-AP that includes $(j,k)$ and no other variables in $C_i$. Since $A$ intersects all $R$-APs, either $k \in A$ on $j \in A$. In the former case, this must be true for every  $k \in A_{i} \setminus A_{i-1}$, and in the latter for every $j \in A_{i-1} \setminus A_i$. Since  $A_i \cap A_{i-1}  \subseteq A$, either $A_{i} \subseteq A$ or $A_{i-1} \subseteq A$.

It remains to be shown that $A_i \not\subseteq A_j$ for any $j \neq i$. Suppose not, so $A_i \subseteq A_j$ for $j \neq i$. Consider $j>i$; similar arguments apply when $j<i$. By Lemma \ref{lem: clique union of intersection}, $l \in A_i \cap A_j$ implies that $l \in [C_{i+1} \cap C_{j+1}] \setminus R(0)$. Since $(C_1,\dots,C_m)$ is the CMCJT, $l \in C_{i+2}$. But then every $l \in A_i$ is also in $ A_{i+1} $, and by Lemma \ref{lem: clique union of intersection}, $C_{i+1} = A_i \cup A_{i+1}  =A_{i+1}\subseteq C_{i+2}$, contradicting that $C_{i+1}$ is maximal.
\end{proof}

We can now prove Lemma \ref{lem: adjacent iff rmap}, establishing the result when $R$ is unconfounded.
\begin{proof}[Proof of Lemma \ref{lem: adjacent iff rmap}]	
Let $(C_1,\dots,C_m)$ by the CMCJT for  $$R'=R\cap  \left[ N^* \left(R\setminus [R(0)\times N] \right) \right]^2,$$  and take $A_i=C_i \cap C_{i+1}$ with $A_0=\{0\}$ and $A_{m}=\{n^*\}$.
By Lemma \ref{lem: separators intersection}, $\mathcal{A}^\rho=\{A_0,\dots,A_m\}$.
	
First, we show that any selection from the $\rho$-separators contains an $R$-MAP.
Let $J$ be a selection from the $\rho$-separators. We construct the $R$-MAP recursively.
For all $i\in J$, let $F(i)=\min\{k:i \in A_k\}$ and $L(i) = \max\{k:i \in A_k\}$. Observe $0 \in J$. Take $i_0=0$, noting $F(0)=L(0)=0$, and let $J_0=J \setminus \{0\}$ be the remaining indexes. We recursively identify $i_l$ from $J_{l-1}=J\setminus \{i_0,\dots, i_{l-1}\}$ as long as $J_{l-1}\neq \emptyset$. To this end, we first need to find the lowest index of minimal separator sets, that does not intersect $\{i_0,\dots,i_{l-1}\}$.
While $J_{l-1} \neq \emptyset$, let $k_l$ be the lowest index for which $J_{l-1} \cap A_{k_l} \neq \emptyset$. 
Note that $k_l$ exists since $J \setminus J_{l-1}$ does not intersect some $A \in \mathcal{A}^\rho$ by definition; $k_l > F(i_{l-1})$ and, when $l>1$, $k_l>L(i_{l-2})+1$ since otherwise $J$ can be made smaller; and $k_l \leq L(i_{l-1})+1$ since otherwise $J \cap A_{L(i_{l-1})+1} = \emptyset$. Moreover, $|J_{l-1} \cap A_{k_l}|=1$.
To see this, suppose not, and there are $j_1,j_2 \in J_{l-1} \cap A_{k_l}$.
Then there is $\kappa_i$ so that $j_i \in A_{k}$ for $k \in [k_l,\kappa_i]$ by Lemma \ref{lem: clique union of intersection} and $(C_1,\dots,C_m)$ is the CMCJT.
If $\kappa_1 \geq \kappa_2$ ($\kappa_2 \geq \kappa_1$), then $J \setminus\{j_2\}$ ($J \setminus\{j_1\}$) still intersects every separator, contradicting minimality of $J$.
	
Define $i_l \in J_{l-1} \cap A_{k_l}$ and $J_l = J_{l-1} \setminus \{i_l\}$. 
By above, $i_{l-1} \in A_{k_l-1}$, so  $i_{l-1} R i_l$ and   $F(i_l)>L(i_{l-2})+1$ implies that $i_{l'} \notr i_l$ for all $l'<l-1$ by Lemma \ref{lem: fun links and cliques}. Moreover, minimality of $J$ requires that $L(i_l)>L(i_{l-1})$. Repeating until $l=|J|-1$, conclude that $(i_0,\dots,i_{|J|-1})$ is an $R$-MAP from $0$ to $n^*$.

Now, we show that the variables in an $R$-MAP are a selection from the $\rho$-separators. Pick any $R$-MAP $(i_0=0,\dots,i_m=n^*)$ and let $J=\{i_0,\dots,i_m\}$.
By Lemma \ref{lem: I contains AP iff block},  $J\cap A \neq \emptyset $ for any separator $A$.
Clearly, $\{i_0 \}= J \cap A_0$ and $\{i_m \}= J \cap A_m$, so $J \setminus \{i_0\}\cap A_0 = \emptyset$ and $J \setminus \{i_m\}\cap A_m = \emptyset$.
By Lemmas \ref{lem: clique union of intersection} and \ref{lem: fun links and cliques}, for each $l \in (0,m)$, there exists $k_l$ so that  $i_l \in A_{k_l}$ and $i_{l+1} \in A_{k_l+1} \setminus A_{k_l}$. Since $i_{l'} \notr i_{l+1}$ for all $l'  \neq l$ by definition, $i_{l'} \notin A_{k_{l}}$ for all $l'  \neq l$ by Lemma \ref{lem: clique union of intersection}. Therefore, $J \setminus \{i_{l}\} \cap A_{k_{l}} = \emptyset$ for all $l\in (0,m)$, and so no proper subset of $J$ intersects every $A \in \mathcal{A}$.
	
(ii) clearly implies (iii). 

[(i) $\implies$ (ii)] Let $(i_0,\dots,i_m)$ be an $R$-MAP from $0$ to $n^*$. By above, $I=\{i_0,\dots,i_m\}$ is a selection from the separators.
Fix $j<m$ and consider $q$ so that 	 $X_{i_j} \perp_q X_{i_{j+1}}$ and $X_I \perp_q X_{I^c}$.
Note that\[q(x_{i_k}|x_{R(i_k)})=
	q(x_{i_k}|x_{R(i_k)\cap I},x_{R(i_k)\cap I^c})=
	q(x_{i_k}|x_{R(i_k)\cap I})=q(x_{i_k}|x_{i_{k-1}})
	\]
for $k>0$.
Therefore, $$q_R(x_{n^*}|a)=\sum_y q(y_{i_1}|a)\prod_{k=1}^{m-2} q(y_{i_{k+1}}|y_{i_{k}})q(x_{n^*}|y_{i_{m-1}}),$$ which does not vary with $a$ if $q(y_{i_{j+1}}|y_{i_j})=q(y_{i_{j+1}})$ for all $y$. 
	
[(iii) $\implies$ (i)] Suppose that  $I$ is a selection from the separators. By the above, $I$ corresponds to an $R$-MAP $(i_0,\dots,i_m)$, and   for any $j$,	 $X_{i_j} \perp_q X_{i_{j+1}}$ and $X_I \perp_q X_{I^c}$ imply that $\rho(a,q)=\rho(b,q)$. We show that $j'$ is $\rho$-adjacent to $k'$ if and only if $\{j',k'\}= \{i_l,i_{l-1}\}$ for some $l$. Pick any distinct $j',k' \in I$.
Suppose first that  $j'=i_l$ and $k'=i_j$ for $j<l-1$. We show that $\rho(a,q)>\rho(b,q)$ for any $q$ with $X_{i_j} \perp_q X_{i_{l}}$, $X_I \perp_q X_{I^c}$, and $\{i',j'\}$-MLRP for all $\{i',j'\} \subset I$ with $\{i',j'\} \neq \{i_j,i_l\},\{0,i_l\}$.
First, $\{0,i_1\}$-MLRP implies that $\marg_{i_1} q(\cdot |a)$ FOSD $\marg_{i_1} q(\cdot |b)$.
For $l'\geq 1$, assume $\marg_{i_{l'}} q_R(\cdot|a)$ FOSD $\marg_{i_{l'}} q_R(\cdot|b)$.
Since $q$ has $\{i_{l'},i_{l'+1}\}$-MLRP,  $\marg_{i_{l'+1}}q(\cdot|x_{i_{l'}})$ FOSD $\marg_{i_{l'+1}} q(\cdot|x'_{i_{l'}})$ when $x>x'$. Therefore, for any event $E=(-\infty,k) \cap X_{i_{l'+1}}$, the function $f=x \mapsto q(E|x_{i_{l'}})$ is decreasing.
Then,
$$q_R(E_{i_{l'+1}}|a)=\sum_{x\in \mathcal{X}_{i_{l'}}} q_R(x_{i_{l'}}|a)f(x) \leq \sum_{x\in \mathcal{X}_{i_{l'}}} q_R(x_{i_{l'}}|b)f(x)=q_R(E_{i_{l'+1}}|b),$$
establishing that $\marg_{i_{l'+1}} q_R(\cdot|a)$ FOSD $\marg_{i_{l'+1}} q_R(\cdot|b)$.
 Inductively extend to $i_m$, and conclude $\rho(a,q)>\rho(b,q)$. Second,  if $j'=i_l$ and $k'=i_{l+1}$, then $X_{j'} \perp_q X_{k'}$ and $X_I \perp X_{I^c}$ implies $\rho(a,q)=\rho(b,q)$ for all $a,b\in S_q$ by the above proof that (i) implies (ii).
 Therefore, $j' \in I$ is $\rho$-adjacent to $k' \in I \setminus\{j'\}$ if and only if $\{j',k'\}$ equals $\{i_{l},i_{l+1}\}$ for some $l$.

It remains to show that such a $q$ exists for any $i,j \in I$.
If $0 \in \{i,j\}$, then relabel so that $i=k=0$; otherwise, relabel so that $k=j$. 
We construct a pseudo-binary dataset $q$ that factorizes as \[
	q(x)=q(x_{i}|x_0)q(x_0)q(x_{j})\prod_{l \in I \setminus \{i,k,0\}} q(x_{l}|x_0,x_{j})  q(x_{I^c})
	\]
where  $q(\bar{E}_{i}|a)>q(\bar{E}_{i}|b)$ [these are $1$ and $0$, respectively, if $i=0$] and for all $l$
$q(\bar{E}_{l}|x_0,x_{j})=q(\bar{E}_{l}|Y)$, for $Y=1$ if and only if $x_0,x_{j} \in \bar{E}_0 \times \bar{E}_{j}$ and otherwise $Y=0$, and $q(\bar{E}_{l}|Y=1)>q(\bar{E}_{l}|Y=0)$.
Clearly $X_{i} \perp_q X_{j}$.
Let $\Delta_l =q(\bar{E}_l|Y=1)-q(\bar{E}_l|Y=0)>0$.

 We claim $\{l,l'\}$-MLRP holds for $\{l,l'\}\neq \{i,j\}, \{k,0\}$.
First, observe that for any $l'\in I \setminus  \{0,i,j\}$, $$q(\underline{E}_{l'}|Y=1)\frac{q(Y=1)}{q(\underline{E}_{l'})}=q(Y=1|\underline{E}_{l'})<q(Y=1|\bar{E}_{l'})=q(\bar{E}_{l'}|Y=1)\frac{q(Y=1)}{q(\bar{E}_{l'})}$$ 
since $q(\bar{E}_{l'}|Y=1)>q(\bar{E}_{l'})$ and $q(\underline{E}_{l'}|Y=1)=1-q(\bar{E}_{l'}|Y=1)<1-q(\bar{E}_{l'})=q(\underline{E}_{l'})$.
Also, for any $l \in \{0,j\}$, $q(Y=1|x_{l})=q(\bar{E}_m)\mathbb{I}_{\bar{E}_{l}}(x_{l})$ for $m \in \{0,j\} \setminus \{l\}$, and  when $l=i$, $q(Y=1|x_{l})=q(\bar{E}_j)q(\bar{E}_{0}|x_{l})$.
In all cases $q(Y=1|x_l)$ is increasing in $x$.
Now, for $l \in I$ and $l' \in I \setminus  \{0,i,j,l\}$, we have $$q(\bar{E}_{l'}|x_{l})=q(Y=1|x_{l})\Delta_{l'}+q(\bar{E}_{l'}|Y=0).$$ Therefore, $q(\bar{E}_{l'}|x_{l})$ weakly increases in $x$, strictly for some $x$.
Similarly, when $l \in \{0,i,j\}$ and  $l' \in   \{0,i,j\} \setminus \{l\}$, either $\{l,l'\}=\{i,j\}$ or $\{l,l'\}=\{0,k\}$  or $\{l,l'\}=\{i,0\}$ and $i \neq 0$. By construction in this case, $q$ has $\{0,i\}$-MLRP. Conclude that $q$ satisfies the conditions.
\end{proof}

It still remains to show necessity for Theorem \ref{thm: RMAP characterize DAG} when $C^*\setminus\{0\}\neq \emptyset$.  Let $N^*$ be the variables in some $R$-MAP from $0$ to $n^*$, i.e. $N^*=  \bigcup_{A \in \mathcal{A}^\rho} A$ and $\hat{N}$ the remaining variables. 
\begin{lemma}\label{lem: not N*}
Let $i^* \in C^* \setminus \{0\}$. If  $(i_0=0,\dots, i_k=n^*)$ and $(j_0=i^*,\dots,j_{k'}=i_k)$ are $R$-MAPs with $\{j_1,\dots,j_{k'-1}\} \subseteq \hat{N} \setminus C^*$ and $k'>1$, then $i_{k-1} R j_{l}$ for all $l \in \{1,\dots,k'-1\}$ and $i^* R i_{l'}$ for all $l' \in \{1,\dots,k-1\}$.  
\end{lemma}
	
\begin{proof}
Note that $i_{k-1} \tilde{R} j_{k'-1}$ as $R$ is perfect, so suppose for contradiction that $j_{k'-1} R i_{k-1}$.
If $k=1$, then $i_{k-1} =0$. Thus, $j_{k'-1} \in C^*$, a contradiction.
If $k>1$, then there exists some $l<k-1$ such that $i_l R j_{k'-1}$ as otherwise $j_{k'-1} \in C^*$. But then $(i_0,\dots, i_l,j_{k'-1}, i_{k-1},i_k)$ is an $R$-MAP, which contradicts that $j_{k'-1}\not\in N^*$. 	
	
Suppose that (*) $i_{k-1}R j_{l'}$ for $l'\in \{l,\dots, k\}$ and, for contradiction, that $i_{k-1}\notr j_{l-1}$.
By perfection, $j_{l-1} R i_{k-1}$. Again, $k-1=0$ implies that $j_{l-1} \in C^*$, a contradiction, so there exists some $l'<k-1$ such that $i_{l'} R j_{l-1}$.
But then $(i_0,\dots, i_{l'},j_{l-1}, \dots, j_{k'-1}, i_k,\dots,i_M)$ is an $R$-MAP, contradicting that $j_{l-1}\notin N^*$. Induction establishes that $i_{k-1} R j_l$ for $l=1,\dots,k'-1$. 

Finally, note that $i^* \tilde{R}i_{k-1}$ by perfection and $i^* R i_{k-1}$  as $R$ is acyclic. The same arguments inductively establish that $i^* R i_{l'}$ for all $l' \in \{1,\dots,k-1\}$.
\end{proof}

\begin{defn}
The covariate $i^* \notin N^*$ is a \emph{revealed confounder} if  
$$X_{ N\setminus \{i^*,0\} } \perp_q X_{  0 }$$
does not imply that $\rho(a,q)=\rho(b,q)$ for all $a,b \in S_q$ and all $q$.
Let $C^*(\rho)$ be the set of all revealed confounders.
\end{defn}

\begin{lemma}
$i^*$ is a $R$-confounder if and only if $i^* \in C^*(\rho)$.
\end{lemma}	

\begin{proof}
Suppose that $ i^* $ is an $R$-confounder. This implies that there is an $R$-MAP $(i_0=i^*,\dots, i_m=n^*)$ that does not go through $0$ or any other $j^*$ such that $j^* R 0$.
Let $I=\{i_0,\dots,i_m,0\}$ and construct a pseudo-binary dataset $q$ that factorizes according to
\[q(x) =  q(x_{i_0}\mid x_{\{i_{1},0\}}) q(x_{0}) q(x_{i_1})\prod_{k'> 1} q(x_{i_{k'}}\mid x_{i_{k'-1}}) \prod_{j\not\in I} q(x_j)\]
where   $q(\bar{E}_{i_1}) \geq \frac{5}{6}$,
$q(\bar{E}_{i_{k'}}\mid \bar{E}_{i_{k'-1}})> q(\bar{E}_{i_{k'}}\mid \underline{E}_{i_{k'-1}})$ for $k'\geq 2$, and
$$q(\bar{E}_{i_0}\mid x_{\{i_{1},0\}})= \begin{cases}
\frac34 \quad &\text{ if } x_{\{i_1,0\}}\in \bar{E}_{i_1}\times \bar{E}_{0},\\
\frac14 & \text{ otherwise}.
\end{cases}$$
By construction, $X_0 \perp_{q} X_{N \setminus \{0,i^*\}}$,
$q(\bar{E}_{i_0}) = \frac14+\frac12 q(\bar{E}_{i_1})q(a) < \frac34$, 
$$   q(\bar{E}_{i_1}\mid \bar{E}_{i_0},a)=\frac{3}{ q(\bar{E}_{i_1})^{-1}+2}   > \frac{1}{ q(\bar{E}_{i_1})^{-1}+2} = q(\bar{E}_{i_1} \mid \underline{E}_{i_0},a),$$ 
and $q(\bar{E}_{i_1} \mid E_{i_0}, b) =q(\bar{E}_{i_1})$.
Hence,  \[
q_{R^*}(\bar{E}_{i_1}\mid a)< \frac34 1 + \frac14 \frac{1}{ 3}=\frac{5}{6}
\leq q(\bar{E}_{i_1})= q(\bar{E}_{i_1} \mid b) .\] 
Combined with the above, we have $q_{R^*}(\bar{E}_{n^*}\mid b) > q_{R^*}(\bar{E}_{n^*}\mid a)$.
Thus, $\rho(b,q)>\rho(a,q)$ and $i^*\in C^*(\rho)$.

Now suppose that $j^*\in \hat{N}$ is not an $R$-confounder. 
WLOG, assume $R= R \cap N^*(R)^2$ so the CMCJT for $R$ given by $(C_1,\dots, C_m)$ is such that $C^*\subset C_1$ and $(C_1,\dots,C_m)$ satisfies the running intersection property with $C_k \cap \bar{C}_{k-1} \subset C_{k-1}$ when  $\bar{C}_{k}=\cup_{k'=1}^{k} C_{k'}$ for every $k$.

Consider any $q$ that satisfies $X_{N\setminus \{0,j^*\}} \perp_q X_0$.
By Lemma \ref{lem: JT rearrange} we have 
\begin{align}
q_{R^*}(x_{\bar{C}_{k}})&= q(x_0) q(x_{C^*\setminus \{0\}}) q(x_{C_1\setminus C^*}\mid x_{C^*}) \prod_{k'=2}^{k} q(x_{C_{k'}\setminus C_{k'-1}}\mid x_{C_{k'}\cap C_{k'-1}}) \nonumber\\
&= q(x_{C_k\setminus C_{k-1}}\mid x_{C_{k}\cap C_{k-1}}) q_{R^*}(x_{\bar{C}_{k-1}}). \label{eq: JT formula k}
\end{align}
We use this to show  that $q_{R^*}(C_k)=q(C_k)$ for every $k$.
For $k=1$, this follows directly from the assumptions on $q$ and $j^*\not \in C^*$: $q_{R^*}(C_1)=q(x_0)q(x_{C^*\setminus\{0\}})q(x_{C_1 \setminus C^*}|x_{C^*})=q(x_0)q(x_{C^*\setminus \{0\}}\mid x_0)q(x_{C_1 \setminus C^*}|x_{C^*})=q(x_{C_1})$.
Assume (IH) that $q_{R^*}(x_{C_{k-1}})=q_{R^*}(x_{C_{k-1}})$. Then,
\begin{align*}
q_{R^*}(x_{C_k}) &= \sum_{y\in \mathcal{X}_{C_{k-1}\setminus C_{k}}}q(x_{C_k\setminus C_{k-1}}\mid x_{C_{k}\cap C_{k-1}}) q_{R^*}(x_{C_{k-1}\cap C_k},y_{C_{k-1}\setminus C_k})\\
&= \sum_{y\in \mathcal{X}_{C_{k-1}\setminus C_k}}q(x_{C_k\setminus C_{k-1}}\mid x_{C_{k}\cap C_{k-1}}) q(x_{C_{k-1}\cap C_k},y_{C_{k-1}\setminus C_k}) = q(x_{C_k}),
\end{align*}
where the second equality follows from (IH). 

Let $k^*$ be the last index for which $0 \in C_{k^*}$.
If $n^* \in C_{k^*}$, then $q_{R^*}(x_{n^*}|a)=q_{R^*}(x_{n^*}|b)$ for all $a,b \in S_q$.
If $j^* \notin C_{k^*} \cap C_{k^*+1}$, then $q(x_{C_{k^*} \cap C_{k^*+1}}|a)=q(x_{C_{k^*} \cap C_{k^*+1}}|b)$ since   $X_{C_{k^*} \cap C_{k^*+1}} \perp_q X_0$.
Then, $q_{R^*}(x_{n^*}|a)=q_{R^*}(x_{n^*}|b)$ for all $a,b \in S_q$ follows from Eq (\ref{eq: JT formula k}).
Therefore, $j^*\not \in C^*(\rho)$ follows if we show $j^* \in C_{k^*} \cap C_{k^*+1}$ or $n^* \in C_{k^*}$.
Suppose not, so $n^* \notin C_{k^*}$ and $j^* \in C_{k^*} \cap C_{k^*+1}$.
By Lemma \ref{lem: fun links and cliques}, $0 R j^*$. For any $l \in C_{k^*+2} \setminus C_{k^*+1}$, $0 R j^* R l$, $0 \notr l$, and there is an $R$-AP from $l$ to $n^*$.
But this implies that  there exists an $R$-MAP $(0,j^*,\dots,n^*)$, contradicting that $j^* \notin N^*$.
\end{proof}
	
\begin{defn} \label{def: revealed confounding path}
There is a \emph{revealed confounding path} from  $i^* \in C^*(\rho) $  to $i_k$  in the $R$-MAP $(i_0=0,\dots,i_M=n^*)$ if
\begin{align*}
& X_j \perp_q X_{N\setminus j} \quad \forall j \in N^* \cup C^*(\rho) \setminus \{i^*,i_0,\dots,i_M\}\\
\&\  & X_{i_k}\perp_q X_{i_{k-1}} 
\end{align*}
do not imply that $\rho(a,q) = \rho(b,q)$.
\end{defn}
	
First condition implies that the only $R$-confounder that matters is $i^*$, and the only $R$-MAP from $0$ to $n^*$ that matters is $(i_0,\dots, i_M)$. The remaining condition ensures that in the absence of a path from $i^*$ to $i_k$, the DM will express indifference.	

\begin{lemma} 
\label{lem: confounding path}	
For $R$-MAP $(i_0=0,\dots,i_M=n^*)$, there is a revealed confounding path from $i^* \in \hat{N}$ to $i_k$ $\iff$ there exists an $R$-MAP from $i^*$ to $i_k$ that does not intersect $N^* \cup C^*$.
\end{lemma}

\begin{proof}
Suppose that $i^*$ is an $R$-confounder and there exists an $R$-MAP from $i^*$ to $i_k$, $(j_0=i^*,j_1,\dots,j_m=i_k)$. that doesn't intersect $N^*$. 
	
Let $C = \{i_0,\dots,i_M,j_0,\dots,j_m\}$. Pick a pseudo-binary dataset $q$ that factorizes according to 
\[q(x)=q(x_{\{i_0,j_0\}})\prod_{k'\neq k}  q(x_{i_{k'}}|x_{i_{k'-1}})\prod_{k'=1}^{m-1} q(x_{j_{k'}}|x_{j_{k'-1}}) q(x_{i_{k}}|x_{i_{k-1},j_{m-1}})
\prod_{j' \notin C} q(x_{j'})\]
where
\begin{enumerate}[(i)]
	\item $q(\bar{E}_{i_{k'}}|\bar{E}_{i_{{k'}-1}})>q(\bar{E}_{i_{k'}}|\underline{E}_{i_{{k'}-1}})$ for $k' \neq  k$;
	\item $q(\bar{E}_{j_{k'}}|\bar{E}_{j_{{k'}-1}})>q(\bar{E}_{j_{k'}}|\underline{E}_{j_{{k'}-1}})$ for $k'<m$; and
	\item for appropriate $z,z',\epsilon\in (0,1)$, 
	\[\begin{array}{c|c||c|c}
X_{\{i_0,j_0\}}& q(.)  &X_{\{i_{k-1},j_{m-1}\}}&  q(\bar{E}_{i_k}|\cdot)\\
\hline
\bar{E}_{i_0}\times\bar{E}_{j_0}&z\frac14 & \bar{E}_{i_{k-1}}\times\bar{E}_{j_{m-1}} & \frac{q\left(\bar{E}_{j_{m-1}}|\underline{E}_{i_{k-1}}\right)}{q\left(\bar{E}_{j_{m-1}}|\bar{E}_{i_{k-1}} \right)}z'+\epsilon\\
\underline{E}_{i_{0}}\times\bar{E}_{j_0}& (1-z) \frac34&   \underline{E}_{i_{k-1}}\times\bar{E}_{j_{m-1}}&z'+\epsilon\\
\bar{E}_{i_0}\times\underline{E}_{j_{0}}&z \frac34 &  \bar{E}_{i_{k-1}}\times\underline{E}_{j_{m-1}}&\epsilon \\
\underline{E}_{i_{0}}\times \underline{E}_{j_{0}} &(1-z)\frac14&   \underline{E}_{i_{k-1}}\times \underline{E}_{j_{m-1}} &\epsilon
\end{array}\]
so $X_{i_{k-1}} \perp_q X_{i_{k}}$ and $q(\bar{E}_{i_k}|x_{j_{m-1}},\bar{E}_{i_{k-1}})>q(\bar{E}_{i_k}|x_{j_{m-1}},\underline{E}_{i_{k-1}}) $ for all $x$.%
\footnote{The quantity $\frac{q(\bar{E}_{j_{m-1}}|\underline{E}_{i_{k-1}})}{q(\bar{E}_{j_{m-1}}|\underline{E}_{i_{k-1}})}$ is pinned down by (1), (2), and the first part of (3), so $z',\epsilon \in (0,1)$ exist satisfying the conditions so that $q$ is a well-defined conditional probability.}
\end{enumerate}
Then, $q_{R^*}(x_C)$ equals
\begin{align*}
&q(x_{i_0})q(x_{j_0})\prod_{k'=1}^{m-1} q(x_{j_{k'}}|x_{R(j_{k'})})\prod_{k'=1}^M q(x_{i_{k'}}|x_{R(i_{k'})})\\
=&q(x_{i_0})q(x_{j_0})\prod_{k'=1}^{m-1}   q(x_{j_{k'}}|x_{j_{k'-1}},x_{i_{k-1}}) \prod_{k'=1}^{k-1}  q(x_{i_{k'}}|x_{i_{k'-1}},x_{j_0})\prod_{k'=k+1}^M q(x_{i_{k'}}|x_{i_{k'-1}}) q(x_{i_k}|x_{i_{k-1}},x_{j_{m-1}})\\
=&q(x_{i_0})q(x_{j_0})\prod_{k'=1}^{m-1} q(x_{j_{k'}}|x_{j_{k'-1}})\prod_{k'=1}^{k-1} q(x_{i_{k'}}|x_{i_{k'-1}})\prod_{k'=k+1}^M q(x_{i_{k'}}|x_{i_{k'-1}}) q(x_{i_k}|x_{i_{k-1}},x_{j_{m-1}})
\end{align*}
where the second equality follows from the construction of $q$.

A simple induction argument establishes that $q(\bar{E}_{i_{k'}}|\bar{E}_{i_{k''}})>q(\bar{E}_{i_{k'}}|\underline{E}_{i_{k''}})$ whenever $k>k'>k''$ or $k'>k''\geq k$, and $q(\bar{E}_{j_{m-1}}|\bar{E}_{j_1})>q(\bar{E}_{j_{m-1}}|\underline{E}_{j_1})$.
Then, $q_{R^*}(\bar{E}_{i_k}|c)$ equals
$$\sum_x q_{R^*}(x_{j_{m-1}}) \left\{q(\bar{E}_{i_{k-1}}|c) \left[q(\bar{E}_{i_k}|\bar{E}_{i_{k-1}},x_{j_{m-1}})-q(\bar{E}_{i_k}|\underline{E}_{i_{k-1}},x_{j_{m-1}}) \right]+q(\bar{E}_{i_k}|\underline{E}_{i_{k-1}},x_{j_{m-1}}) \right\}$$
and
$$q_{R^*}(\bar{E}_{n^*}|c)=q_{R^*}(\bar{E}_{i_k}|c)[q(\bar{E}_{n^*}|\bar{E}_{i_k})-q(\bar{E}_{n^*}|\underline{E}_{i_k})]+q(\bar{E}_{n^*}|\underline{E}_{i_k}).$$
Combining implies $q_{R^*}(\bar{E}_{n^*}|a)>q_{R^*}(\bar{E}_{n^*}|b)$.

Now, suppose that there is no path in $\hat{N}$ from $i^*$ to $i_k$.
Observe that if  there is a path in $\hat{N}$ from $i^*$ to $i_{k+1}$, there is also a path to $i_{k}$ by Lemma \ref{lem: not N*}.
Consider any $q$ satisfying the conditions.
Define $R'$ that removes all edges involving a node in $N^* \cup C^*(\rho) \setminus \{i^*,i_0,\dots,i_M\}$ but otherwise agree with $R$.
Take $R''=R' \cap N^*(R')^2$, and arguments as in Lemma \ref{lem: I contains AP iff block} show that $q_{R^*}(x_{n^*}|c)=q_{R'{}^*}(x_{n^*}|c)=q_{R''{}^*}(x_{n^*}|c)$.
Lemma \ref{lem: fun links and cliques} shows that $N^*(R')$ equals the variables that are part of some $R'$-MAP from $i^*$ or $0$ to $n^*$.
Let $l \in R''(i_k)$.
Either $l$ is in an $R''$-MAP from $0$ to $n^*$, in which case $l=i_{k-1}$, or $l$ is in an  $R''$-MAP from $i^*$ to $n^*$, in which case $l=i_{k-1}$ or $l$ is part of a path in $\hat{N}$ from $i^*$ to $i_k$,  a contadiction. 
Therefore, $R''(i_{k})=\{i_{k-1}\}$ and $q(x_{i_{k}}|x_{R''(i_{k})})=q(x_{i_{k}}|x_{i_{k-1}})=q(x_{i_k})$ since $X_{i_k} \perp_q X_{i_{k-1}}$.
Then, the DM is indifferent, and there is no revealed confounding path from $i^*$ to $i_k$.
\end{proof}

\begin{defn}
If $(i_0=0,\dots,i_m=n^*)$ is an $R$-MAP and there is a revealed confounding path from $i^* \in C^*(\rho)$  to $i_k$, then $B \subset \{i^* \} \cup [\hat{N} \setminus C^*(\rho)]$ is a $(i^*,i_k)$-separator if $\rho(a,q)=\rho(b,q)$ whenever
\begin{align*}
&X_j \perp_q X_{\{j\}^c} \quad \forall j \in N^* \cup C^*(\rho) \setminus \{i^*,i_0,\dots,i_m\}\\
&X_{i_k}\perp_q X_{i_{k-1}} \\
\&\ & X_B \perp_q X_{B^c}.
\end{align*}
\end{defn}

The first two conditions are the same as before, and the third condition is analogous to separation.

\begin{rem}
$\{i^*\}$ is a  $(i^*,i_k)$-separator but $\{i_k\}$ is not.
\end{rem}
	
\begin{lemma}
If $B \subset \{i^* \} \cup [\hat{N} \setminus C^*(\rho)]$, then $B$ is a $(i^*,i_k)$-separator $\iff$ $B$ intersects all $R$-paths from $i^*$ to $i_k$ in $\hat{N}$.
\end{lemma}
	
\begin{proof}
If $B$ does not intersect a path from $i^*$ to $i_k$ contained in $[(N^* \cup C^*)\setminus \{i^*,i_k\}]^c$, then we can find a minimal $R$-path  $(j_0=i^*,j_1,\dots,j_m=i_k)$ with  $j_l \notin B\cup N^* \cup C^*$ for all $l\neq 0,m$. 
Then, the dataset $q$ from Lemma \ref{lem: confounding path} constructed for this path leads to $\rho(a,q)\neq \rho(b,q)$.
	
If $B$ intersects all paths from $i^*$ to $i_k$, then consider  $R'$ that drops all edges involving a node in $B$ and a node in $B^c$.
Since $X_{B} \perp_q X_{B^c}$, $q_{R'{^*}}(x_{n^*}|x_0)=q_{R^*}(x_{n^*}|x_0)$ for all $x \in S_q \times \mathcal{X}_{-0}$. 
Applying Lemma \ref{lem: confounding path} to $R'$ and $q$ establishes that $\rho(a,q)=\rho(b,q)$. 
\end{proof}

\begin{defn}\label{def: revealed adjacent off path}
Let $\mathcal{A}^{(i^*,i_k)}$ be the set of minimal $(i^*,i_k)$-seperators, and  say that $J$ is a selection from $\mathcal{A}^{(i^*,i_k)}$ if $J \cap A \neq \emptyset$ for all $A \in \mathcal{A}^{(i^*,i_k)}$ and no proper subset of $J$ has a non-empty intersection with every $A \in \mathcal{A}^{(i^*,i_k)}$.
For any selection $J$ from $\mathcal{A}^{(i^*,i_k)}$,  $i \in J$ is \emph{adjacent} to $j \in J$ if $\rho(a,q)=\rho(b,q)$ for all $a,b \in S_q$ whenever
\begin{align*}
& X_l \perp_q X_{N\setminus l} \quad \forall l \in N \setminus [\{ i_0,\dots,i_m\} \cup J],\\
& X_{i} \perp_q X_{j}|X_{i_{k-1}}, \text{ and}\\
\&\ & X_{i_k}\perp_q X_{i_{k-1}}.
\end{align*}
\end{defn}
 The first two conditions are analogous to the conditions in Lemma \ref{lem: adjacent iff rmap}.(ii), and the third condition ensure that the DM expresses indifference whenever all variables outside the $R$-MAP between $0$ and $n^*$ are independent. 

\begin{lemma}
\label{lem: adjacent off path}
The sequence $(j_0=i^*,\dots,j_{M-1},i_k)$ is an $R$-MAP  from $i^*\in C^*(\rho)$ to $i_k$ and  $J=\{j_0, \dots,j_{M-1}\}$ does not intersect $N^*$ if and only if $J$ is a selection from $\mathcal{A}^{(i^*,i_k)}$ and $j_l$ is adjacent to $j_{l+1}$  for all $l<M-2$.
\end{lemma}

\begin{proof}	
As in Lemma \ref{lem: adjacent iff rmap}, $J$ is a selection from $\mathcal{A}^{(i^*,i_k)}$ if and only if there is an $R$-MAP $(j_0=i^*,\dots,j_M=i_k)$, $\{j_0,\dots,j_{M-1}\}=J$, and $J\cap C^*(\rho)=\{i^*\}$.

Let $N^\dag=\{ i_0,\dots,i_m\} \cup J$.
We show that $j_{k'}$ is adjacent to $j_{{k'}+1}$.
Pick any $q$ satisfying the conditions for $j_{k'}$ and $j_{k'+1}$ to be adjacent.
Then by Lemma \ref{lem: not N*}, $$q(x_{i_l}|x_{R(i_l)})=q(x_{i_l}|x_{R(i_l)\cap N^\dag},x_{R(i_l) \setminus N^\dag})=q(x_{i_l}|x_{R(i_l)\cap N^\dag})=q(x_{i_l}|x_{i_{l-1}},x_{j_0})$$
for $l<k$.
Similarly, $q(x_{i_l}|x_{R(i_l)})=q(x_{i_l}|x_{i_{l-1}})$ for $l>k$. Moreover,
$$q(x_{j_l}|x_{R(j_l)})=q(x_{j_l}|x_{R(j_l)\cap N^\dag},x_{R(j_l) \setminus N^\dag})=q(x_{j_l}|x_{R(j_l)\cap N^\dag})=q(x_{j_l}|x_{j_{l-1}},x_{i_{k-1}})$$
for $l>0$. In particular,
$q(x_{j_{{k'}+1}}|x_{R(j_{{k'}+1})})=q(x_{j_{{k'}+1}}|x_{i_{k-1}})$ by the second condition in Definition \ref{def: revealed adjacent off path}.

Note that for $s=k'+2,...,M  $,
\begin{align*}
 &\sum_{y_{j_{s-1}}} 
q(x_{j_{s}}|y_{j_{s-1}},x_{i_{k-1}})q(y_{j_{ s-1}}|x_{i_{k-1}})=\sum_{y_{j_{ s-1}}} 
q(y_{j_{ s-1}},x_{j_{ s }}|x_{i_{k-1}})\\
=& 
\sum_{y_{j_{ s-1}}} q(y_{j_{ s-1}}|x_{j_{ s }},x_{i_{k-1}})q(x_{j_{ s }}|x_{i_{k-1}})=q(x_{j_{ s }}|x_{i_{k-1}}).
\end{align*}
Since $j_{s }$ is the only variable dependent on $j_{s-1}$  in  $q_{R^*}(\cdot |x_0)$   and $q(x_{j_{{k'}+1}}|x_{R(j_{{k'}+1})})=q(x_{j_{{k'}+1}}|x_{i_{k-1}})$, we can successively apply the above to drop $x_{j_{s-1}}$  for $s=k'+2,...,M $ from the expression for
$q_{R^*}(x_{n^*}|x_0)$.
Since $X_{i_k} \perp_q X_{i_{k-1}}$, $$\ q (x_{j_M}|x_{i_{k-1}} ) =q(x_{i_k}|x_{i_{k-1}})=q(x_{i_k}) $$
so the DM is indifferent, and  $j_{k'}$ and $j_{{k'}+1}$ are adjacent.
	
We now show that $j_{l}$ is not adjacent to $j_{k^*+1}$ for any $l<k^*$.
Pick a pseudo-binary dataset $q$ that factorizes as 
\begin{align*}
q(x)=&q(x_{\{i_0,j_0\}})q(j_{k^*+1})\prod_{k'\neq k}  q(x_{i_{k'}}|x_{i_{k'-1}})\prod_{k'\notin\{k^*,k^*+1\}} q(x_{j_{k'}}|x_{j_{k'-1}}) \prod_{j' \notin N^\dag} q(x_{j'})\times\\
&\qquad\times q(j_{k^*}|j_{k^*+1},j_{k^*-1}) q(i_{k}|i_{k-1},j_{M-1})
\end{align*}
and where
$$q(\bar{E}_{i_{k}}|x_{\{i_{k-1},j_{M-1}\}} )= \left\{
\begin{array}{ll}
\frac34 & \text{if }x_{\{i_{k-1},j_{m-1}\}} \in \bar{E}_{i_{k-1}}\times \bar{E}_{j_{M-1}}\\
\frac14 & \text{otherwise}
\end{array} \right.	,$$	
$$q(\bar{E}_{j_{k^*}}|x_{\{j_{k^*-1},j_{k^*+1}\}} )= \left\{
\begin{array}{ll}
\frac34 & \text{if }x_{\{j_{k^*-1},j_{k^*+1}\}} \in \bar{E}_{j_{k^*-1}}\times \bar{E}_{j_{k^*+1}}\\
\frac14 & \text{otherwise}
\end{array} \right.	,$$
$q(\bar{E}_{j_{k'}}|\bar{E}_{j_{k'-1}})>q(\bar{E}_{j_{k'}}|\underline{E}_{j_{k'-1}})$ for $k' \neq k^*,k^*+1$,
and $q(\bar{E}_{i_{k'}}|\bar{E}_{i_{k'-1}})>q(\bar{E}_{i_{k'}}|\underline{E}_{i_{k'-1}})$ for $k' \neq k$.
By construction, $X_l \perp_q X_{k^*+1} |X_{i_{k-1}}$ for all $l<k^*$.
We can also calculate that $q(\bar{E}_{j_{k'}}|\bar{E}_{k'-1})>q(\bar{E}_{j_{k'}}|\underline{E}_{j_{k'-1}})$ for $k'=k^*,k^*+1$, so $q_{R^*}(\bar{E}_{j_m}|\bar{E}_{j_0})>q_{R^*}(\bar{E}_{j_m}|\underline{E}_{j_0})$. As in Lemma \ref{lem: adjacent iff rmap}, this leads to $\rho(a,q)>\rho(b,q)$. Conclude that $j_{k^*+1}$ and $j_l$ are not adjacent for all $l<k^*$.
\end{proof}

\subsection{Proof of Proposition \ref{prop: MAP nongeneric}}
Any  dataset $q$ where $X_I \perp_q X_{I^c}$ can be parametrized by  $$(\alpha,\beta) \in \left\{(x,y)\in \mathbb{R}^{m}_+ \times \mathbb{R}^k_+: \sum x_i= 1,\sum y_i = 1 \right\}=\Delta$$ where $m= \left| \prod_{j\in I} \mathcal{X}_j  \right|
$ and $k=\left| \prod_{j\notin I} \mathcal{X}_j  \right|
$. Enumerating $\prod_{j\in I} \mathcal{X}_j =\{\omega^I_1,\dots,\omega^I_m\}$ and $\prod_{j\notin I} \mathcal{X}_j =\{\omega^c_1,\dots,\omega^c_k\}$, $q(\omega^I_i,\omega^c_j)=\alpha_i \beta_j$. $\Delta$ can be thought of as a positive measure subset of $\mathbb{R}^{(m-1)(k-1)}$.

Pick a DAG $Q$. Then, let $f_Q:\Delta \rightarrow [0,1]$ equal $\sum_x u(x) [q_Q(x_{n^*}|a)-q_Q(x_{n^*}|b)]$ when $q$ corresponds to $\Delta$. $f$ is a sum of ratios of polynomials in $\Delta$, as \[
q_Q(x_{n^*}|c)=\sum_{y \in \mathcal{X}:y_0=0,y_{n^*}=x_{n^*}}\prod \frac{ q(y_{j},y_{Q(j)})}{q(y_{Q(j)})}\]
and $$q(y_{B})=\sum_{j,k:(\omega_j,\omega_k)_{B} = y_{B} } \alpha_j \beta_k$$ for any $B \subset N$. The sum of a ratio of polynomials is a ratio of polynomials.

By Lemma \ref{lem: adjacent iff rmap}, there exists an $(\alpha,\beta)$ so that $f_Q(\alpha,\beta) \neq f_R(\alpha,\beta)$ whenever $(i_0,\dots,i_m)$ is not a $Q$-MAP. 
Therefore, $f_Q - f_R \neq 0$ and by standard results, e.g. Caron and Traynor (2005), the set $\{z \in \Delta: f_Q(z)-f_R(z)=0\}$ has zero measure. Therefore, the set of datasets for which $\sum u(x)[q_Q(x_{n^*}|a)-q_Q(x_{n^*}|b)] \neq \sum u(x)[ q_R(x_{n^*}|a)-q_R(x_{n^*}|b)]$ is open, dense, and has full measure on $\Delta$. \hfill \qedsymbol

\subsection{Necessity for Theorem \ref{thm: proper SCR}} 
Suppose that $\rho$ has an Endogenous SCR $(R,u,n^*)$. The $\rho$ is also represented by the DAG $R^\rho$ that eliminates links between unnecessary variables in $R$, i.e. $R^\rho =R \cap N(R)^2$. Lemma \ref{claim_fundamental} implies that $R^\rho$ corresponds to a CMCJT $(C_1,\dots,C_m)$, and Lemma \ref{lem: clique union of intersection} and \ref{lem: I contains AP iff block} imply that $\mathcal{A}^\rho$ can be labeled $(A_1,\dots,A_{m+1})$ so that $C_i=A_i \cup A_{i+1}$. By Lemma \ref{lem: adjacent iff rmap}, every $j \in A_i \setminus A_{i+1}$ is adjacent to every $k \in A_{i+1}\setminus A_i$. 
Therefore, Axiom \ref{ax: Consistent Revealed Causes} holds.
Since   $$\int u(c)d\rho^S_{R^\rho}(c_{n^*}|a')\in \left[ \min_{x \in \mathcal{X}_{n^*}}  u(x),\max_{x \in \mathcal{X}_{n^*}}  u(x) \right],$$
Axioms \ref{ax: first SCR} and \ref{ax: bounded} hold. 
Axioms    \ref{ax: same pr if first clique}, \ref{ax: luce axiom inference}, and \ref{ax: IA R markov} 
follow from Lemma \ref{lem: JT rearrange} and continuity. \hfill \qedsymbol

\subsection{Sufficiency for Theorem \ref{thm: proper SCR}}
Suppose that $\rho$ satisfies the axioms.
Axiom \ref{ax: Consistent Revealed Causes} implies that an adjacency ordering of $\mathcal{A}^\rho$,  $(A_1,\dots,A_{k})$, exists. Let $\{n^*\}=A_k$. Define $R$ by $j R k$ if and only if there exists $i$ so that either $k \in A_{i+1} \setminus A_i$ and $j \in A_i$ or $j,k \in A_{i+1} \setminus A_i$ and $j<k$. $R$ is acyclic by CRC.(i) since each $j \in A_{i+1}\setminus A_i$ for exactly one $i$, call it $i_j$, and $j R k$ only if $i_j \leq i_k$. If $j R l$ and $k R l$, then $i_l \geq i_j,i_k$, and, after relabeling, $i_j \leq i_k$. If $i_j=i_k$, then $j \tilde{R} k$. Otherwise, $i_j<i_k \leq i_l$ and $j \in A_{i_l-1}$ by definition of $R$.   Then, $j R k$ since $j \in A_{i_k-1}$ by CRC.(ii). Therefore, $R$  is a perfect, unconfounded, nontrivial DAG.

We show that
\begin{equation}
\label{eq: luce ratior rep}
\frac{\rho(a,S)}{\rho(b,S)}=\frac{\exp[\int_{\mathcal{X}_{n^*}} u(c)d\rho^S_{R}(c_{n^*}|a)]}{\exp[\int_{\mathcal{X}_{n^*}} u(c)d\rho^S_{R}(c_{n^*}|b)]}
\end{equation}
for any $a,b \in S$ and any $S\in \mathcal{S}$. If so, then  $\rho$ has  an Endogenous SCR $(R,u,n^*)$  since $\sum_{a\in S} \rho(a,S)=1$.
Pick any $S \in \mathcal{S}$ and any $a,b \in S$. 
By Lemma \ref{lem: JT rearrange}, \[
\rho^S_{R}(x_{N^*(R)})=
\rho^S(x_{A_1})\prod_{i=1}^{k-1} \rho^S \left(x_{A_{i+1}}|x_{A_i} \right).
\]
Let $a' (y)=\rho^S_{R}(y|a )$ and $b' (y)=\rho^S_{R}(y|b )$ for  every $y\in \mathcal{X}_{-0}$.
Since $\{a',b'\}$ is correctly perceived, $\rho(a',\{a',b'\}) /\rho(b',\{a',b'\}) $ has the desired form by Axiom \ref{ax: IA R markov}.
If $a'=b'$, then $\marg_{A^*_1} a= \marg_{A^*_1}b$, so $\rho(a,S)=\rho(b,S)$ by Axiom \ref{ax: same pr if first clique}, and Equation (\ref{eq: luce ratior rep}) holds.

Otherwise, let $S_1=\{a',b'\}$ and recursively define $S_m=S_{m-1} \cup \{ \frac{1}{m} a' + \frac{m-1}{m} b'\}$. Each $S_m$ is correctly perceived by construction, and each has $m+1$ distinct alternatives.
By Axiom \ref{ax: bounded}, there exists $K>0$ so that for any $a'',b'' \in S'' \in \mathcal{S}$, $\frac{\rho(a'',S'')}{\rho(b'',S'')} \leq K$.
For $S_m \cup S =\left\{s_1,\dots,s_{M} \right\}$ and any  $i,j  \in \{1,\dots,M \}$ with $i \neq j$, we have $\rho(s_i,S_m \cup S)\geq K^{-1} \rho(s_j,S_m \cup S)$. Then, \[1 = \sum_{i \neq j} \rho(s_i,S_m \cup S) + \rho(s_j,S_m \cup S) \geq [(M -1) K^{-1} +1]\rho(s_j,S_m \cup S)\]
so $\rho(s_j,S_m \cup S) \leq \frac{K}{M +K-1} $. In particular $\sum_{c\in S} \rho(c,S_m \cup S)\rightarrow 0$ as $m\rightarrow \infty$.

For $p_m= \rho^{S_m\cup S}$ and arbitrary $1<i < |\mathcal{A}^\rho |$ we have $p_m(x_{A_{i+1}}|x_{A_i})$ equals \[\frac{1}{p_m(x_{A_i})}
\left[ \sum_{a''\in S} p_m(a'') p_m(x_{A_i}|a'') a''(x_{A_{i+1}}|x_{A_i})+p_m(S_m)p_m(x_{A_i}|x_0 \in S_m) a'(x_{A_{i+1}}|x_{A_i})
\right]\]
for every $x \in \mathcal{X}_{-0}$ since  $\hat{a}(x_{A_{i+1}}|x_{A_i})=a'(x_{A_{i+1}}|x_{A_i})$ for all $\hat{a} \in S_m$.
This converges to  $\rho^{S_1}(x_{A_{i+1}}|x_{A_i})=a'(x_{A_{i+1}}|x_{A_i})$ because $p_m(a'')\rightarrow 0$ for all $a'' \in S$. Since $i$ was arbitrary, $\rho^{S_m \cup S}(x_{A_{i+1} }|x_{A_i}) \rightarrow \rho^{S_1}(x_{A_{i+1}  }|x_{A_i})$ for every $i$. 

Axiom \ref{ax: same pr if first clique} gives that $\rho(a,S_m \cup S)=\rho(a',S_m \cup S)$ and $\rho(b,S_m \cup S)=\rho(b',S_m \cup S)$. Axiom \ref{ax: luce axiom inference} implies that  $$\frac{\rho(a',S_m \cup S)}{\rho(b',S_m \cup S) }=\frac{\rho(a,S_m \cup S)}{\rho(b,S_m \cup S)}\rightarrow \frac{\rho(a',S_1)}{\rho(b', S_1)}$$ 
and that $$\frac{\rho(a,S_m \cup S)}{\rho(b,S_m \cup S)}=\frac{\rho(a',S_m \cup S)}{\rho(b',S_m \cup S) }\rightarrow\frac{ \rho(a, S)}{\rho(b,S)}.$$ 
Therefore,  $\frac{\rho(a',S_1)}{\rho(b', S_1)}=\frac{ \rho(a, S)}{\rho(b,S)}$ and Equation (\ref{eq: luce ratior rep}) holds for $a,b$. Since $a$, $b$, and $S$ were arbitrary, $\rho$ has an Endogenous SCR $(R,n^*,u)$. \hfill \qedsymbol

\subsection{Proof of Proposition \ref{prop:coarseness}}
Suppose that $\rho_i$ has a perfect SCR $(R_i,u_i,n^*)$ for $i=1,2$ and that $\rho_2$ has a coarser model than $\rho_1$.
 Let $N^*$ be the set of all $i$ that are part of an $R_2$-MAP from $0$ or an $R_2$-confounder to $n^*$.
Define $R_3= R_1 \cap N^*\times N^*$  and $\rho_3$ to have an SCR $(R_3,u,n^*)$.
Pick any $q$, and set $\hat{q}=x \mapsto q(x_{N^*}) \prod_{i \notin N^*} q(x_i)$.
By construction for $i=2,3$, $q_{R_i}(x_{n^*}|do(x_0))= \hat{q}_{R_i}(x_{n^*}|do(x_0))$ so $\arg\max_{S_q} \rho_i(\cdot,q)=\arg\max_{S_q} \rho_i(\cdot,\hat{q})$. Also, $\hat{q}_{R_1}(x_{n^*}|do(x_0))= \hat{q}_{R_3}(x_{n^*}|do(x_0))$ so $\arg\max_{S_q} \rho_1(\cdot,\hat{q})=\arg\max_{S_q} \rho_3(\cdot,\hat{q})$. By hypothesis, $\arg\max_{S_q} \rho_1(\cdot,\hat{q})=\arg\max_{S_q} \rho_2(\cdot,\hat{q})$.
Combining, $\arg\max_{S_q} \rho_2(\cdot,q)=\arg\max_{S_q} \rho_3(\cdot,q)$. Since $q$ was arbitrary, $R_3$ represents $\rho_2$.
%
 
Conversely, let $\rho_i$ have a perfect  SCR $(R_i,u_i,n^*)$ for $i=1,2$, $u_2 = u_1$ and $R_2 = R_1 \cap \left[ N' \times N' \right]$ for some $N' \subset N$.
By Lemmas \ref{claim_fundamental} and \ref{lem: fun links and cliques}, $N^*(R_2)$ equals the indexes that appear an $R_2$-MAP from $0$ or an $R_2$-confounder to $n^*$.
Pick any $q \in \mathcal{D}$ so that $X_i \perp_q X_{N \setminus \{i\}}$ for all $i \notin N^*(R_2)$. Since $N' \supset N^*(R_2)$, we also have $X_i \perp_q X_{N \setminus \{i\}}$ for all $i  \notin  N'$, so for every $i$ and  $x \in \mathcal{X}$,
\[ q\left(x_i|x_{R_1(i)} \right)= q \left( x_i|x_{R_1(i) \cap N'},x_{R_1(i) \setminus N'} \right)= q\left( x_i|x_{R_1(i) \cap N'} \right)=q\left(x_i|x_{R_2(i)} \right). \]
Hence, $q_{R_1}=q_{R_2}$, and so $\rho_1(a,q)\geq \rho_1(b,q)$ for all $b \in S_q$ if and only if $\rho_2(a,q)\geq \rho_2(b,q)$ for all $b \in S_q$.\hfill \qedsymbol

\section{Examples}
\subsection{Violating Regularity}\label{ex: regularity}
If $\rho(a,S)=z$, then 
\[ 
\rho^S(1_{\consequence}|0_{\vari })=0<\rho^S(1_{\consequence}|1_{\vari })
=\frac{1}{1+z},
\]
and since $a(1_{\vari })>b(1_{\vari})>0$, we have $1>z>\frac12$. 
Then, $\frac12<\rho^S(1_{\consequence}|1_{\vari })<\frac23$, so $\rho^S_R(1_{\consequence}|b)<\frac13$, while $\rho^S_R(1_{\consequence}|a)>\frac12$. 
Hence 
\[
\frac{\rho(b,S)}{\rho(a,S)}<\frac{\exp[\frac13 6+\frac23 0]}{\exp[\frac12 6+\frac12 0]}= \exp [-1]<\frac12
\]
and $\rho(b,S)<\frac13$.
Because profit is independent of size of workforce conditional on own output according to $R_{\vari }$, the consultant is indifferent between two workforce sizes which lead to the same distribution of own output.
Therefore, $\rho(c,S')=\rho(b,S')=\gamma$. Then,
\[ 
\rho^{S'}(1_{\consequence}|1_{\vari })=\frac{(1-2\gamma) \frac12+ \gamma \frac12 + \gamma(0)}{(1-2\gamma)+ 2 \gamma \frac12}=\frac12=
\frac{(1-2\gamma) (0) +\gamma(0) + \gamma\frac12}{(1-2\gamma)(0)+2\gamma\frac12 }=\rho^{S'}(1_{\consequence}|0_{\vari }),
\]
so $\rho(a,S')=\rho(c,S')=\rho(b,S')=\frac13>\rho(b,S)$, violating regularity.

\subsection{Self-confirming choices}
\label{ex: selfconfirming}One can compute that
\[
\rho^S \left(1_{\consequence}|1_{\vari} \right)=\frac{2-2\rho(a,S)}{2-\rho(a,S)}\text{ and }\rho^S \left( 1_{\consequence}|0_{\vari} \right)= q.
\] 
That is, whether the consultant thinks that own output has a positive or negative effect on profit depends on the fraction who take action $a$.
Setting $u(1)=\lambda >u(0)=0$,  
$$
\lambda \frac12 \left( q- \frac{2-2\rho(a,S) }{2-\rho(a,S)}  \right) = \ln \rho(a,S) - \ln (1-\rho(a,S)).$$
For large enough $\lambda$, $\rho(a,S) \approx 0$ and $\rho(a,S) \approx 1$ are both solutions.

\bibliographystyle{plainnat}
\bibliography{subjectivecausality}

\begin{thebibliography}{57}
\providecommand{\natexlab}[1]{#1}
\providecommand{\url}[1]{\texttt{#1}}
\expandafter\ifx\csname urlstyle\endcsname\relax
  \providecommand{\doi}[1]{doi: #1}\else
  \providecommand{\doi}{doi: \begingroup \urlstyle{rm}\Url}\fi

\bibitem[Alexander and Gilboa(2019)]{Alexander}
Yotam Alexander and Itzhak Gilboa.
\newblock Subjective causality.
\newblock September 2019.

\bibitem[Andre et~al.(2022)Andre, Pizzinelli, Roth, and Wohlfart]{andre2021}
Peter Andre, Carlo Pizzinelli, Christopher Roth, and Johannes Wohlfart.
\newblock Subjective models of the macroeconomy: Evidence from experts and a
  representative sample.
\newblock \emph{Review of Economic Studies}, Forthcoming, 2022.

\bibitem[Apesteguia and Ballester(2018)]{apesteguia2018monotone}
Jose Apesteguia and Miguel~A. Ballester.
\newblock {Monotone Stochastic Choice Models: The Case of Risk and Time
  Preferences}.
\newblock \emph{Journal of Political Economy}, 126\penalty0 (1):\penalty0
  74--106, 2018.

\bibitem[Bohren and Hauser(2021)]{BohrenHauser2018}
J.~A. Bohren and D.~Hauser.
\newblock Learning with heterogenous misspecified models: Charaterization and
  robustness.
\newblock \emph{Econometrica}, 89\penalty0 (6):\penalty0 3025--3077, 2021.

\bibitem[Bohren et~al.(2023)Bohren, Haggag, Imas, and Pope]{Bohren2023}
J~Aislinn Bohren, Kareem Haggag, Alex Imas, and Devin~G Pope.
\newblock Inaccurate statistical discrimination: An identification problem.
\newblock \emph{Review of Economics and Statistics}, pages 1--45, 2023.

\bibitem[Brady and Rehbeck(2016)]{brady2016menu}
Richard~L Brady and John Rehbeck.
\newblock Menu-dependent stochastic feasibility.
\newblock \emph{Econometrica}, 84\penalty0 (3):\penalty0 1203--1223, 2016.

\bibitem[Card(1999)]{Card1999}
David Card.
\newblock The causal effect of education on earnings.
\newblock volume~3 of \emph{Handbook of Labor Economics}, pages 1801--1863.
  Elsevier, 1999.

\bibitem[Cattaneo et~al.(2020)Cattaneo, Ma, Masatlioglu, and
  Suleymanov]{cattaneo2020randomattention}
Matias Cattaneo, Xinwei Ma, Yusufcan Masatlioglu, and Elchin Suleymanov.
\newblock A random attention model.
\newblock \emph{Journal of Political Economy}, 128, 2020.

\bibitem[{Cerreia Viogolio} et~al.(2022){Cerreia Viogolio}, Hansen, Maccheroni,
  and Marinacci]{Cerreiaetal2017ModelMisspecification}
Simone {Cerreia Viogolio}, Lars~Peter Hansen, Fabio Maccheroni, and Massimo
  Marinacci.
\newblock Making decisions under model misspecification.
\newblock \emph{working paper}, 2022.

\bibitem[Chambers et~al.(2022)Chambers, Cuhadaroglu, and
  Masatlioglu]{chambersetal2021influence}
Christopher~P Chambers, Tugce Cuhadaroglu, and Yusufcan Masatlioglu.
\newblock {Behavioral Influence}.
\newblock \emph{Journal of the European Economic Association}, 05 2022.
\newblock ISSN 1542-4766.
\newblock \doi{10.1093/jeea/jvac028}.
\newblock URL \url{https://doi.org/10.1093/jeea/jvac028}.
\newblock jvac028.

\bibitem[Cowell et~al.(1999)Cowell, Dawid, Lauritzen, and
  Spiegelhalter]{cowell1999}
R.~Cowell, P.~Dawid, S.~Lauritzen, and D.~Spiegelhalter.
\newblock \emph{Probabilistic Networks and Expert Systems}.
\newblock Springer, 1999.

\bibitem[Denrell(2018)]{denrell2020sampling}
Jerker Denrell.
\newblock Sampling biases explain decision biases.
\newblock pages 49--95. Oxford University Press, 2018.

\bibitem[Eliaz and Spiegler(2020)]{Eliaz2018model}
Kfir Eliaz and Ran Spiegler.
\newblock A model of competing narratives.
\newblock \emph{American Economic Review}, 110\penalty0 (12):\penalty0
  3786--3816, 2020.

\bibitem[Eliaz et~al.(2020)Eliaz, Spiegler, and Weiss]{Eliaz2020}
Kfir Eliaz, Ran Spiegler, and Yair Weiss.
\newblock Cheating with models.
\newblock \emph{American Economic Review: Insights}, 2020.

\bibitem[Eliaz et~al.(2021)Eliaz, Spiegler, and Thysen]{Eliaz2019}
Kfir Eliaz, Ran Spiegler, and Heidi~C Thysen.
\newblock Persuasion with endogenous misspecified beliefs.
\newblock \emph{European Economic Review}, 134:\penalty0 103712, 2021.

\bibitem[Eliaz et~al.(2022)Eliaz, Galperti, and Spiegler]{eliaz2022false}
Kfir Eliaz, Simone Galperti, and Ran Spiegler.
\newblock False narratives and political mobilization.
\newblock \emph{arXiv preprint arXiv:2206.12621}, 2022.

\bibitem[Ellis and Masatlioglu(2022)]{EllisMasatlioglu2021}
Andrew Ellis and Yusufcan Masatlioglu.
\newblock Choice with endogenous categorization.
\newblock \emph{The Review of Economic Studies}, 89\penalty0 (1):\penalty0
  240--278, 2022.

\bibitem[Ellis and Piccione(2017)]{EllisPiccione2017}
Andrew Ellis and Michele Piccione.
\newblock Correlation misperception in choice.
\newblock \emph{American Economic Review}, 107\penalty0 (4):\penalty0 1264--92,
  April 2017.

\bibitem[Esponda(2008)]{Esponda2008}
Ignacio Esponda.
\newblock Behavioral equilibrium in economies with adverse selection.
\newblock \emph{American Economic Review}, 98\penalty0 (4):\penalty0 1269--91,
  September 2008.

\bibitem[Esponda and Pouzo(2016)]{EspondaPouzo2016Berk}
Ignacio Esponda and Demian Pouzo.
\newblock Berk–nash equilibrium: A framework for modeling agents with
  misspecified models.
\newblock \emph{Econometrica}, 84\penalty0 (3):\penalty0 1093--1130, 2016.
\newblock \doi{https://doi.org/10.3982/ECTA12609}.
\newblock URL \url{https://onlinelibrary.wiley.com/doi/abs/10.3982/ECTA12609}.

\bibitem[Esponda and Vespa(2018)]{esponda2018endogenous}
Ignacio Esponda and Emanuel Vespa.
\newblock Endogenous sample selection: A laboratory study.
\newblock \emph{Quantitative Economics}, 9\penalty0 (1):\penalty0 183--216,
  2018.

\bibitem[Eyster and Rabin(2005)]{EysterRabin2005}
Erik Eyster and Matthew Rabin.
\newblock Cursed equilibrium.
\newblock \emph{Econometrica}, 73\penalty0 (5):\penalty0 1623--1672, 2005.

\bibitem[Frick et~al.(2020)Frick, Iijima, and Ishii]{Fricketal2019}
Mira Frick, Ryota Iijima, and Yuhta Ishii.
\newblock Misinterpreting others and the fragility of social learning.
\newblock \emph{Econometrica}, 88\penalty0 (6):\penalty0 2281--2328, 2020.

\bibitem[Gul and Pesendorfer(2006)]{Gul2006Random}
Faruk Gul and Wolfgang Pesendorfer.
\newblock Random expected utility.
\newblock \emph{Econometrica}, 74\penalty0 (1):\penalty0 121--146, 2006.

\bibitem[Hajek et~al.(1992)Hajek, Havranek, and Jirousek]{Hajeketal1992}
Petr Hajek, Tomas Havranek, and Radim Jirousek.
\newblock \emph{Uncertain Information Processing in Expert Systems}.
\newblock CRC Press, 1992.

\bibitem[He(2022)]{He2018}
Kevin He.
\newblock Mislearning from censored data: The gambler's fallacy and other
  correlational mistakes in optimal-stopping problems.
\newblock \emph{Theoretical Economics}, 17\penalty0 (3):\penalty0 1269--1312,
  2022.

\bibitem[Heidhues et~al.(2018)Heidhues, Koszegi, and
  Strack]{HeidhuesKoszegiStrack2018}
P.~Heidhues, B.~Koszegi, and P.~Strack.
\newblock Unrealistic expectations and misguided learning.
\newblock \emph{Econometrica}, 86\penalty0 (4):\penalty0 1159--1214, 2018.

\bibitem[Imbens(2020)]{Imbens2020}
Guido~W Imbens.
\newblock Potential outcome and directed acyclic graph approaches to causality:
  Relevance for empirical practice in economics.
\newblock \emph{Journal of Economic Literature}, 58\penalty0 (4):\penalty0
  1129--79, 2020.

\bibitem[Jehiel and Koessler(2008)]{Jehiel2008}
Philippe Jehiel and Fr{\'e}d{\'e}ric Koessler.
\newblock Revisiting games of incomplete information with analogy-based
  expectations.
\newblock \emph{Games and Economic Behavior}, 62\penalty0 (2):\penalty0
  533--557, 2008.

\bibitem[Ke et~al.(2020)Ke, Zhao, Wang, and Hsieh]{Keetal2020}
Shaowei Ke, Chen Zhao, Zhaoran Wang, and Sung-Lin Hsieh.
\newblock Behavioral neural networks.
\newblock \emph{Working paper}, 2020.

\bibitem[Kochov(2018)]{Kochov2015}
Asen Kochov.
\newblock A behavioral definition of unforeseen contingencies.
\newblock \emph{Journal of Economic Theory}, 175:\penalty0 265--290, 2018.

\bibitem[K\"{o}szegi and Rabin(2006)]{KR06}
B.~K\"{o}szegi and M.~Rabin.
\newblock A model of reference-dependent preferences.
\newblock \emph{Quarterly Journal of Economics}, 121\penalty0 (4):\penalty0
  1133--1165, 2006.

\bibitem[Lang and Kahn-Lang~Spitzer(2020)]{Lang2020}
Kevin Lang and Ariella Kahn-Lang~Spitzer.
\newblock Race discrimination: An economic perspective.
\newblock \emph{Journal of Economic Perspectives}, 34\penalty0 (2):\penalty0
  68--89, 2020.

\bibitem[Langer(1975)]{Langer1975illusion}
E.~Langer.
\newblock The illusion of control.
\newblock \emph{Journal of Personality and Social Psychology}, 32:\penalty0
  311--328, 1975.

\bibitem[Levy et~al.(2022)Levy, Razin, and Young]{Levyetal2021}
Gilat Levy, Ronny Razin, and Alwyn Young.
\newblock Misspecified politics and the recurrence of populism.
\newblock \emph{American Economic Review}, 112\penalty0 (3):\penalty0 928--62,
  2022.

\bibitem[Lipman(1999)]{Lipman1999}
Barton~L. Lipman.
\newblock Decision theory without logical omniscience: Toward an axiomatic
  framework for bounded rationality.
\newblock \emph{The Review of Economic Studies}, 66\penalty0 (2):\penalty0 pp.
  339--361, 1999.

\bibitem[Lombrozo(2007)]{lombrozo2007simplicity}
Tania Lombrozo.
\newblock Simplicity and probability in causal explanation.
\newblock \emph{Cognitive psychology}, 55\penalty0 (3):\penalty0 232--257,
  2007.

\bibitem[Lu(2016)]{Lu2016}
Jay Lu.
\newblock Random choice and private information.
\newblock \emph{Econometrica}, 84\penalty0 (6):\penalty0 1983--2027, 2016.

\bibitem[Luce(1959)]{luce1959individual}
R~Duncan Luce.
\newblock Individual choice behavior.
\newblock 1959.

\bibitem[Manzini and Mariotti(2014)]{manzini2014stochastic}
Paola Manzini and Marco Mariotti.
\newblock Stochastic choice and consideration sets.
\newblock \emph{Econometrica}, 82\penalty0 (3):\penalty0 1153--1176, 2014.

\bibitem[{Montiel Olea} et~al.(2021){Montiel Olea}, Ortoleva, Pai, and
  Prat]{olea2021competing}
Jose~Luis {Montiel Olea}, Pietro Ortoleva, Mallesh~M Pai, and Andrea Prat.
\newblock Competing models.
\newblock \emph{working paper}, 2021.

\bibitem[Pacer and Lombrozo(2017)]{pacer2017ockham}
M~Pacer and Tania Lombrozo.
\newblock Ockham’s razor cuts to the root: Simplicity in causal explanation.
\newblock \emph{Journal of Experimental Psychology: General}, 146\penalty0
  (12):\penalty0 1761, 2017.

\bibitem[Pearl(2009)]{Pearl2009}
J.~Pearl.
\newblock \emph{Causality: Models, Reasoning and Inference}.
\newblock Cambridge University Press, 2009.

\bibitem[Pearl(1995)]{Pearl1995Causal}
Judea Pearl.
\newblock Causal diagrams for empirical research.
\newblock \emph{Biometrika}, 82\penalty0 (4):\penalty0 669--688, 1995.

\bibitem[Pearl and Verma(1991)]{Pearl1991}
Judea Pearl and T.~S. Verma.
\newblock A theory of inferred causation, 1991.

\bibitem[Samuelson and Mailath(2020)]{SammuelsonMailath2019}
L.~Samuelson and G.~Mailath.
\newblock Learning under diverse world views: Model based inference.
\newblock \emph{American Economic Review}, 110\penalty0 (5):\penalty0 1464 --
  1501, May 2020.

\bibitem[Samuelson and Zeckhauser(1988)]{SamuelsonZeckhauser88}
W.~Samuelson and R.~Zeckhauser.
\newblock Status quo bias in decision making.
\newblock \emph{Journal of Risk and Uncertainty}, 1:\penalty0 7--59, 1988.

\bibitem[Schenone(2020)]{Schenone2020causality}
Pablo Schenone.
\newblock Causality: A decision theoretic foundation.
\newblock Technical report, 2020.

\bibitem[Schumacher and Thysen(2022)]{Schumacher}
Heiner Schumacher and Heidi~Christina Thysen.
\newblock Equilibrium contracts and boundedly rational expectations.
\newblock \emph{Theoretical Economics}, 17\penalty0 (1):\penalty0 371--414,
  2022.

\bibitem[Shermer(1998)]{Shermer1998}
Martin Shermer.
\newblock \emph{Why people believe weird things: pseudoscience, superstition,
  and other confusions of our time}.
\newblock Freeman \& Co, 1998.

\bibitem[Sloman(2005)]{Sloman2005}
Steven Sloman.
\newblock \emph{Causal Models: How People Think about the World and Its
  Alternatives}.
\newblock Oxford University Press, 2005.

\bibitem[Spiegler(2016)]{Spiegler2016}
Ran Spiegler.
\newblock Bayesian networks and boundedly rational expectations.
\newblock \emph{The Quarterly Journal of Economics}, 131\penalty0 (3):\penalty0
  1243--1290, 2016.

\bibitem[Spiegler(2017)]{Spiegler2017}
Ran Spiegler.
\newblock “data monkeys”: a procedural model of extrapolation from partial
  statistics.
\newblock \emph{The Review of Economic Studies}, 84\penalty0 (4):\penalty0
  1818--1841, 2017.

\bibitem[Spiegler(2020)]{Spiegler2020a}
Ran Spiegler.
\newblock Can agents with causal misperceptions be systematically fooled?
\newblock \emph{Journal of the European Economic Association}, 18\penalty0
  (2):\penalty0 583--617, 2020.

\bibitem[Tennant et~al.(2020)Tennant, Murray, Arnold, Berrie, Fox, Gadd,
  Harrison, Keeble, Ranker, Textor, Tomova, Gilthorpe, and
  Ellison]{Tennat2020DAGHealth}
Peter W~G Tennant, Eleanor~J Murray, Kellyn~F Arnold, Laurie Berrie, Matthew~P
  Fox, Sarah~C Gadd, Wendy~J Harrison, Claire Keeble, Lynsie~R Ranker, Johannes
  Textor, Georgia~D Tomova, Mark~S Gilthorpe, and George T~H Ellison.
\newblock {Use of directed acyclic graphs (DAGs) to identify confounders in
  applied health research: review and recommendations}.
\newblock \emph{International Journal of Epidemiology}, 50\penalty0
  (2):\penalty0 620--632, 12 2020.
\newblock ISSN 0300-5771.
\newblock \doi{10.1093/ije/dyaa213}.
\newblock URL \url{https://doi.org/10.1093/ije/dyaa213}.

\bibitem[Verma and Pearl(1991)]{Verma1991equivalence}
T.~S. Verma and Judea Pearl.
\newblock Equivalence and synthesis of causal models.
\newblock Technical report, 1991.

\bibitem[Wason(1960)]{Wason1960Congruence}
P.~C. Wason.
\newblock On the failure to eliminate hypotheses in a conceptual task.
\newblock \emph{Quarterly Journal of Experimental Psychology}, 12\penalty0
  (3):\penalty0 129--140, 1960.

\end{thebibliography}

\end{document}